\documentclass[12pt]{article}

\usepackage[margin=1in]{geometry}

\usepackage{setspace}

\usepackage{graphicx}
\usepackage{float}

\usepackage{enumitem}

\usepackage{microtype}
\usepackage{subfigure}
\usepackage{booktabs}
\usepackage{comment}
\usepackage{wrapfig}
\usepackage{arydshln}
\usepackage{enumitem}
\usepackage{multirow}
\usepackage{url}
\usepackage{natbib}
\usepackage{authblk}

\RequirePackage[colorlinks,citecolor=blue,linkcolor=blue,urlcolor=blue,pagebackref]{hyperref}

\usepackage{amsmath}
\usepackage{amssymb}
\usepackage{mathtools}
\usepackage{amsfonts}
\usepackage{bm}
\usepackage{bbm}

\usepackage{algorithm}
\usepackage{algorithmic}

\def\1{\bm{1}}

\DeclareMathAlphabet{\mathsfit}{\encodingdefault}{\sfdefault}{m}{sl}
\SetMathAlphabet{\mathsfit}{bold}{\encodingdefault}{\sfdefault}{bx}{n}

\def\bbV{{\mathbb{V}}}

\newcommand{\E}{\mathbb{E}}

\newcommand{\R}{\mathbb{R}}

\DeclareMathOperator*{\argmin}{arg\,min}

\newcommand{\p}[1]{\left(#1\right)}
\newcommand{\sqb}[1]{\left[#1\right]}

\newcommand{\abs}[1]{\left|#1\right|}

\newcommand{\norm}[1]{\left\|#1\right\|}

\mathtoolsset{showonlyrefs=true}

\newif\ifECMA
\makeatletter
\@ifclassloaded{econsocart}{\ECMAtrue}{\ECMAfalse}
\makeatother

\usepackage{amsthm}

\theoremstyle{plain}

\ifECMA
\newtheorem{theorem}{Theorem}
\else
\newtheorem{theorem}{Theorem}[section]
\fi
\newtheorem{lemma}[theorem]{Lemma}

\newtheorem{corollary}[theorem]{Corollary}
\newtheorem{proposition}[theorem]{Proposition}

\newtheorem*{remark}{Remark}

\usepackage[textsize=tiny]{todonotes}
\usepackage{multirow}
\usepackage{wrapfig}
\usepackage{subfigure}

\usepackage{xcolor}
\newcount\Comments  
\newcommand{\kibitz}[2]{\ifnum\Comments=1\textcolor{#1}{#2}\fi}

\usepackage{listings}

\usepackage{tabularx}

\newcommand{\trans}{^{\prime}}
\usepackage{threeparttable}

\ifECMA
\else
\usepackage{float}
\usepackage{epstopdf}
\fi

\allowdisplaybreaks

\title{Counterexamples and Sufficient Conditions: Comments on ``Optimally-Transported Generalized Method of Moments''}

\ifECMA
\author[add1,add2,add3,add4]{\fnms{Masahiro}~\snm{Kato}\ead[label=e1]{mkato-csecon@g.ecc.u-tokyo.ac.jp}}
\address[add1]{\orgname{The University of Tokyo}}
\address[add2]{\orgname{Mizuho-DL Financial Technology}}
\address[add3]{\orgname{RIKEN AIP}}
\address[add4]{\orgname{Osaka Metropolitan University}}
\else
\author{Masahiro Kato\thanks{Email: \texttt{mkato-csecon@g.ecc.u-tokyo.ac.jp}}$\,$}
\affil{The University of Tokyo}
\affil{Mizuho-DL Financial Technology}
\affil{RIKEN AIP}
\affil{Osaka Metropolitan University}
\date{\today}
\fi

\begin{document}

\maketitle 

\noeqref{eq:Phi-def,eq:S-def,eq:S-integrability,eq:bm-bound,eq:coordinate-density,eq:first-multiplier-limits,eq:foc-z,eq:newton,eq:observed-linear,eq:observed-mean,eq:primal,eq:q-second-component,eq:sample-global-ineq,eq:scalar-moment,eq:second-component-map,eq:second-component-sample-solution,eq:second-component-transport-cost,eq:stationary-factor,eq:two-component-foc-z,eq:continuous-covariance,eq:continuous-jacobian}

\begin{abstract}
We comment on the optimally-transported generalized method of moments (OTGMM) of \citet{Schennach2026optimallytransported} and give counterexamples to Theorems~2--6. In the small-error analysis, the assumptions do not ensure consistency or asymptotic normality (Theorems~2--3). In the large-error analysis, Theorem~4 states the equivalence between the OTGMM estimator and the GMM estimator with modified moments, and Theorems~5--6 assert consistency and asymptotic normality. Through counterexamples, we show that the equivalence in Theorem~4 fails and a component of the Lagrange multiplier has different limits under the two estimators (Theorems~5--6). We also point out that, under misspecification, the population value selected by OTGMM depends on the metric. We give sufficient conditions and establish results following the original motivation for OTGMM. Our corrections to Theorems~4--6 change their assumptions or interpretation more substantially than those to Theorems~2 and~3. Furthermore, we correct Assumption~16 and the test's null distribution.
\end{abstract}
\ifECMA
\end{frontmatter}
\else
\fi

\section{Introduction}
\label{sec:introduction}
\citet{Schennach2026optimallytransported} proposes optimally-transported GMM (OTGMM), a method combining GMM with optimal transport. OTGMM is mainly motivated by overidentified moment models, especially when an overidentification test rejects the null that the moment conditions hold. One approach to this problem is the generalized empirical likelihood (GEL) method, which changes the weights assigned to the observations. OTGMM instead maps each observation $x_i$ to a value $z_i$ and chooses $\theta$ and $z_1,\ldots,z_n$ so that the sample moment conditions hold while minimizing the average squared distance between $z_i$ and $x_i$.

In this study, we give counterexamples to Theorems~2--6 in \citet{Schennach2026optimallytransported} under their stated assumptions and then establish corrected results under sufficient conditions. In particular, we examine two questions: whether Theorems~2--6 describe the global minimizer of the sample version of the problem that defines the OTGMM estimator, and which parameter minimizes the population transport cost when the observed distribution does not satisfy the moment conditions. We show that the assumptions in the original study do not justify the claimed results for the OTGMM estimator. The problems are less severe for Theorems~2 and~3 than for Theorems~4--6. For Theorems~2 and~3, additional sufficient conditions restore consistency and asymptotic normality of the parameter component of the OTGMM estimator. However, the equivalence asserted in Theorem~4 can fail, and the asymptotic analysis in Theorems~5 and~6 concerns the GMM estimator with modified moments, which need not solve the original constrained problem.

\subsection{Counterexamples and Technical Issues}
Theorems~2 and~3 concern the small-error analysis of OTGMM, where the differences $z_i-x_i$ between adjusted and observed values are small. As a counterexample, we construct a smooth overidentified model satisfying Assumptions~2--11 in which the OTGMM estimator is not consistent. Along a subsequence, every global minimizer of the sample version of the problem lies outside a fixed neighborhood of the true parameter with probability bounded away from zero. On events whose probabilities remain bounded away from zero along this subsequence, changing at most two observations by amounts that tend to zero makes the sample moment conditions hold at a fixed incorrect parameter value, with transport cost lower than the minimum cost at any parameter in a fixed neighborhood of the truth. The stated assumptions are therefore insufficient for consistency in Theorem~2, and hence for the centered normal limit in Theorem~3.

Theorems~4--6 provide the large-error analysis, which does not rely on the small-error approximation. We first show a counterexample to Theorem~4 and then show a separate counterexample to Theorems~5 and~6. Theorem~4 states that solving the original constrained problem is equivalent to solving a just-identified GMM system with modified moments. These moments are formed by a selection rule: for each observation, the rule takes the solution nearest to the observation among the solutions of the first-order condition of the OTGMM Lagrangian with respect to the adjusted value. Our counterexample to Theorem~4 is a scalar model in which the constrained problem has a unique global minimizer, but at the multiplier associated with this minimizer, the rule selects a different value that violates the sample moment condition. The modified moment equations therefore do not characterize the OTGMM minimizer, and the equivalence in Theorem~4 fails. Note that the counterexample to Theorems~5 and~6 also provides a counterexample to Theorem~4.

In addition to the counterexample to Theorem~4, we show a separate counterexample to Theorems~5 and~6. We construct an overidentified model in which the observed vector has a density and the assumptions of these theorems hold, yet the first component of the Lagrange multiplier has different probability limits under the OTGMM estimator and under the GMM estimator with modified moments. Because the asymptotic analysis in these theorems is based on the modified moment equations, the different probability limits show that this analysis does not establish the stated joint asymptotic results for the OTGMM estimator defined by the original constrained problem.

We also raise several technical issues regarding the population value selected by OTGMM, Assumptions~13 and~16, and the test of the absence of error. We first explain the issue concerning the population value selected by OTGMM. Even when a global minimizer of the original constrained problem is obtained, the population value selected by OTGMM need not equal the parameter in a latent-variable model. If the observed distribution satisfies the moment conditions at a unique parameter and the population optimization problem has a minimizer at each parameter value, OTGMM selects that parameter with zero adjustment. By contrast, if no parameter satisfies them, the selected value depends on the transport metric and on which variables the estimator may adjust. A simple errors-in-variables regression illustrates the difference between the value selected by OTGMM and the coefficient in the latent-variable model. With only the usual moment condition for ordinary least squares (OLS), OTGMM leaves the data unchanged and returns the attenuated OLS estimator. Optimal transport alone therefore does not recover the coefficient on the latent regressor; doing so requires additional moment conditions and assumptions linking the observed and latent variables.

We then examine the assumptions used in the large-error analysis. With linearly dependent moment conditions, Assumption~13 does not imply the uniqueness required in Assumption~12. Assumption~16 does not imply the matrix bound used in the proofs in the supplement to the original study, so we replace it with a condition on the matrix that appears in that bound.

After Theorem~6, \citet{Schennach2026optimallytransported} proposes a formal test of the absence of error, with the null hypothesis that the Lagrange multiplier is zero. However, that theorem does not justify the proposed asymptotic null distribution of $\chi^2_{d_g}$ in equation~(14) of the original study.

\subsection{Corrected Results}
For Theorems~2 and~3 in the small-error analysis of \citet{Schennach2026optimallytransported}, we restore consistency of the parameter component of the OTGMM estimator by imposing a uniform bound on how much the moments can change when the data are adjusted. The bound allows moment functions that are quadratic in the adjusted values, such as linear instrumental-variable moments. Our asymptotic-normality result requires that the parameter components of all global minimizers lie within $O_p(n^{-1/2})$ of the truth and that a local expansion of the sample moments and a quadratic approximation to the minimized transport cost hold uniformly throughout that neighborhood.

For Theorem~4 in the large-error analysis of \citet{Schennach2026optimallytransported}, we give a sufficient condition under which the selected adjusted values minimize the transport cost subject to the sample moment conditions at a given parameter value. When the moment functions are affine in the adjusted observations, their coefficient matrix has full row rank, and the adjusted observations have no support restrictions, this condition holds and the adjustments have a closed form. To obtain a global OTGMM minimizer, the resulting transport cost must also be minimized over $\theta$. For Theorems~5 and~6, we state sufficient conditions for consistency and asymptotic normality of the GMM estimator with modified moments after specifying which solution of the first-order condition of the OTGMM Lagrangian with respect to the adjusted value is used to form those moments.

We correct the existence claim following Problem~(1) in \citet{Schennach2026optimallytransported}. A nonempty constraint set does not ensure that a population minimizer exists: we give an example in which the problem is feasible but the infimum is not attained. We also give a sufficient condition for the population optimization problem to have a minimizer at a given parameter value.

For the test of the absence of error, we derive a statistic from expansions of the sample first-order conditions around the true parameter and a zero multiplier. Its asymptotic null distribution is $\chi^2_{d_g-d_\theta}$, where $d_g$ is the dimension of the moment vector and $d_\theta$ that of the parameter.

We independently reproduce the OTGMM coefficient on log highway kilometers in the two main empirical specifications. This confirms the reported coefficients, but it does not establish that the reproduced estimates attain the global minimum of the OTGMM objective function. For the specification in Table~I of \citet{Schennach2026optimallytransported}, we also show that the constrained problem used in the application, with the dependent variable and constant fixed, has a strict local minimum that is not a global minimum. At this local minimum, the highway coefficient is close to the reproduced value, whereas several coefficients on the controls change signs.

\subsection{Organization}
Section~\ref{sec:recap} reviews the OTGMM objective function and the results studied here. Section~\ref{sec:small-error} gives the counterexamples to Theorems~2 and~3, and Section~\ref{sec:large-problems} those to Theorems~4--6. We then discuss the population value selected by OTGMM in Section~\ref{sec:parameter} and present corrected results in Section~\ref{sec:corrected}. Although the problems differ in severity, we present the counterexamples and the corrected results in the order of the theorems in the original study for readability. The remaining corrections and the numerical calculations appear in Section~\ref{sec:additional}, with the conclusion in Section~\ref{sec:conclusion}. The appendices give the main proofs. Additional counterexamples, longer calculations, and replication details are in the supplemental appendix.

\section{OTGMM and the Results Studied Here}
\label{sec:recap}
\subsection{Setup}
Let $\mu_x$ denote the distribution of the observed random vector $x$. Let $\mathcal X\subset\R^{d_x}$, let $\Theta\subset\R^{d_\theta}$, and let $g:\mathcal X\times\Theta\to\R^{d_g}$, where $d_x$, $d_\theta$, and $d_g$ are the dimensions of $x$, $\theta$, and $g$, respectively. Throughout, $\norm{\cdot}$ denotes the Euclidean norm, the superscript $\trans$ denotes transpose, and $\bbV$ denotes the variance operator. For a probability measure $\mu_{zx}$ on $\mathcal X\times\mathcal X$, let $\E_{\mu_{zx}}$ denote expectation under $\mu_{zx}$. When the relevant probability measure is clear, we write $\E$ and $\Pr$.

For each $\theta\in\Theta$, we write the population optimization problem in \citet{Schennach2026optimallytransported}, with its objective multiplied by $1/2$, as
\begin{align}
Q(\theta)
=
\inf_{\mu_{zx}}
\left\{
\frac{1}{2}\E_{\mu_{zx}}[\norm{z-x}^2]:
\E_{\mu_{zx}}[g(z,\theta)]=0
\right\},
\label{eq:population-problem}
\end{align}
where the infimum is over probability measures $\mu_{zx}$ supported on $\mathcal X\times\mathcal X$ and having marginal $\mu_x$ for $x$. The parameter values selected by the population optimization problem are the minimizers of $Q(\theta)$ over $\theta\in\Theta$.

For observations $x_1,\ldots,x_n\in\mathcal X$, let $z_i\in\mathcal X$ denote the value to which $x_i$ is mapped, and for any function $a$, define $\widehat{\E}[a(z,x)]=n^{-1}\sum_{i=1}^n a(z_i,x_i)$. The sample version of the problem in \eqref{eq:population-problem}, which defines the OTGMM estimator, is
\begin{align}
\min_{\theta,z_1,\ldots,z_n}
\frac{1}{2}\widehat{\E}[\norm{z-x}^2]
\quad\text{subject to}\quad
\widehat{\E}[g(z,\theta)]=0.
\label{eq:primal}
\end{align}
For each $\theta\in\Theta$, let $\widehat Q(\theta)$ denote the infimum of $\frac{1}{2}\widehat{\E}[\norm{z-x}^2]$ over $z_1,\ldots,z_n$ satisfying $\widehat{\E}[g(z,\theta)]=0$. By a global minimizer of \eqref{eq:primal}, we mean a feasible tuple $(\theta,z_1,\ldots,z_n)$ that attains the minimum over the entire feasible set. Let $\widehat\Theta_n$ denote the set of parameter components of all such global minimizers. A local minimizer or a point that merely satisfies the first-order conditions may fail to attain the global minimum, in which case it does not solve the sample version of the problem.

Unless another limit is stated, all probability limits and all $O_p$ and $o_p$ statements below are taken as $n\to\infty$.

\subsection{Theorems~2 and~3 in the Small-Error Analysis} 
Theorems~2 and~3 analyze cases in which the differences $z_i-x_i$ become small as $n\to\infty$. Theorem~2 states that the OTGMM estimator is consistent for the true parameter $\theta_0$ and that the Lagrange multiplier satisfies $\lambda=O_p(n^{-1/2})$. Theorem~3 gives a Gaussian limit centered at $\theta_0$. Both results require the global minimum of \eqref{eq:primal} to be attained near $\theta_0$.

\subsection{Theorems~4--6 in the Large-Error Analysis} 
Theorems~4--6 concern the general case in which the small-error approximation is not imposed. Let $\lambda\in\R^{d_g}$ denote the Lagrange multiplier for the sample moment condition in \eqref{eq:primal}. The analysis in Theorems~4--6 starts from the Lagrangian
\begin{align}
\widehat{\E}\left[
\frac{1}{2}\norm{z-x}^2-\lambda\trans g(z,\theta)
\right].
\label{eq:sample-lagrangian}
\end{align}
For a vector argument $v$, let $\partial_v$ denote differentiation with respect to $v$, and let $\partial_{v\trans}$ denote the Jacobian with respect to $v\trans$. Following \citet{Schennach2026optimallytransported}, define $H(z,\theta)=\partial_{z\trans}g(z,\theta)$. The first-order conditions of the OTGMM Lagrangian with respect to $\theta$, $\lambda$, and $z_i$ are, respectively,
\begin{align}
\widehat{\E}[\partial_\theta g\trans(z,\theta)]\lambda&=0,
\label{eq:foc-theta}\\
\widehat{\E}[g(z,\theta)]&=0,
\label{eq:foc-lambda}\\
(z_i-x_i)-H(z_i,\theta)\trans\lambda&=0,\qquad i=1,\ldots,n.
\label{eq:foc-z}
\end{align}
For given values of $(x_i,\theta,\lambda)$, condition~\eqref{eq:foc-z} may have more than one solution in $z_i$. Theorem~4 defines $q(x_i,\theta,\lambda)$ as the solution nearest to $x_i$:
\begin{align}
q(x,\theta,\lambda)
=
\argmin_{z\in\mathcal X:\,z-H(z,\theta)\trans\lambda=x}
\norm{z-x}^2.
\label{eq:q-nearest}
\end{align}
The theorem then defines the modified moment function
\begin{align}
\widetilde g(x,\theta,\lambda)
=
\begin{pmatrix}
\partial_\theta g\trans(q(x,\theta,\lambda),\theta)\lambda\\
g(q(x,\theta,\lambda),\theta)
\end{pmatrix}.
\label{eq:modified-moments}
\end{align}
Following Theorem~4, write $\widetilde\theta=(\theta\trans,\lambda\trans)\trans$ and $\widetilde g(x,\widetilde\theta)=\widetilde g(x,\theta,\lambda)$, and refer to $\widehat{\E}[\widetilde g(x,\widetilde\theta)]=0$ as the sample moment equations based on $\widetilde g$. Theorem~4 states that solving these equations is equivalent to solving \eqref{eq:primal}, and on this basis Theorems~5 and~6 state consistency and asymptotic normality of the OTGMM estimator. We call the estimator defined by these equations the \emph{GMM estimator with modified moments}. We use the superscripts $\text{prime}$ and $\text{modified}$ for quantities obtained from a global minimizer of \eqref{eq:primal} and from the GMM estimator with modified moments, respectively.

\section{Counterexamples to Theorems~2 and~3}
\label{sec:small-error}

The assumptions of Theorems~2 and~3 control the moment functions at the observed data and impose derivative conditions at the true parameter. They do not uniformly bound how much $g(z,\theta)$ can change when $z$ is close to an observation and $\theta$ is far from the truth. The next proposition gives a counterexample to these two theorems. We take $\mathcal X=\R^2$.

\begin{proposition}[Counterexample to Theorems~2 and~3]\label{prop:small-error}
There exists a model with connected compact parameter space \(\Theta=[0,2]\) and interior true value \(\theta_0=3/2\) that satisfies Assumptions~2--11 of \citet{Schennach2026optimallytransported}. In this model, \(x_i\in\R^2\) has a strictly positive infinitely differentiable density, and \(g:\R^2\times\Theta\to\R^2\) extends to an infinitely differentiable function on an open neighborhood of its domain. Along a deterministic subsequence \((n_m)_{m\geq1}\) with \(n_m\to\infty\), we have
\begin{align}
\liminf_{m\to\infty}
\Pr\left(
\inf_{\theta\in\widehat\Theta_{n_m}}
\abs{\theta-\theta_0}>0.9
\right)>0.
\label{eq:small-inconsistency}
\end{align}
Consequently, the consistency conclusion of Theorem~2 in \citet{Schennach2026optimallytransported} and the centered normal limit in Theorem~3 do not hold under their stated assumptions.
\end{proposition}

Appendix~\ref{app:small-error} gives the construction and verifies Assumptions~2--11. The integer $m$ indexes the mixture components used in the construction, and $n_m$ denotes the sample size used for the $m$th mixture component. For each coordinate $r=1,2$, the probability that the sample contains an observation whose $r$th coordinate is drawn from the $m$th mixture component remains bounded away from zero as $m\to\infty$. On the events used in the proof, whose probabilities also remain bounded away from zero, changing one observed coordinate value for each $r=1,2$ makes the two sample moment conditions hold at $\theta^*=1/2$, which is separated from $\theta_0$. Each change converges to zero in absolute value as $m\to\infty$, and the transport cost is $o(1/n_m)$. On these events, every parameter outside a fixed neighborhood of $\theta^*$ has minimum cost bounded below by a positive multiple of $1/n_m$, so for all sufficiently large $m$, every global minimizer lies near $\theta^*$ and remains separated from $\theta_0$. Theorems~2 and~3 do not list Assumption~1, which presupposes a solution of the population optimization problem at each $\theta\in\Theta$, and their proofs do not use it. In this model, for every $\theta$ in a neighborhood of $\theta^*$, the infimum in \eqref{eq:population-problem} is zero but is not attained.

The proof of Theorem~2 in \citet[pp.~636--637]{Schennach2026optimallytransported} uses a mean-value expansion and the Cauchy--Schwarz inequality to argue that, when $\theta$ is bounded away from $\theta_0$, the sample moment condition cannot be satisfied by values $z_i$ whose average squared distance from $x_i$ converges to zero as $n\to\infty$. Proposition~\ref{prop:small-error} shows why the stated assumptions do not justify this argument. At $\theta_0$ and $z=x$, the derivative of the moment function with respect to $z$ equals the identity matrix, whereas at $\theta^*$ it can be arbitrarily large in the regions used in the construction. A change in $z$ whose norm converges to zero as $m\to\infty$ can therefore produce an order-one change in a sample moment.

A uniform bound on changes in $g(z,\theta)$ when $z$ is close to $x$ would suffice for this argument. Section~\ref{sec:corrected} gives such a condition and treats consistency of the parameter components of global minimizers separately from the quadratic approximation to the minimized transport cost used for asymptotic normality.

\section{Counterexamples to Theorems~4--6}
\label{sec:large-problems}

For a given $(x,\theta,\lambda)$, condition~\eqref{eq:foc-z} may have more than one solution in $z$. Theorem~4 defines $q(x,\theta,\lambda)$ by minimizing $\norm{z-x}^2$ over those solutions. This selection does not ensure that the resulting value minimizes the corresponding term of the Lagrangian. Define $\ell(z;x,\theta,\lambda)$ by
\begin{align}
\ell(z;x,\theta,\lambda)
=
\frac{1}{2}\norm{z-x}^2-\lambda\trans g(z,\theta).
\label{eq:one-observation-lagrangian}
\end{align}
This function is the contribution of one observation to the Lagrangian in \eqref{eq:sample-lagrangian}. Two solutions of \eqref{eq:foc-z} can give different values of $\lambda\trans g(z,\theta)$, so the solution nearest to $x$ need not minimize $\ell$. Values chosen separately for the observations solve the constrained problem only if they jointly satisfy the sample moment condition. The function $q$ and the sample moment equations based on $\widetilde g$ therefore impose the three first-order conditions on the selected values $z_i=q(x_i,\theta,\lambda)$, but they do not guarantee equivalence with the original constrained problem. Uniqueness of the solution to those equations is not enough because the parameter and multiplier associated with a global minimizer need not satisfy them.

The counterexample to Theorem~4 is a scalar model in which the constrained problem has a unique global minimizer, whose $z$ component also globally minimizes the Lagrangian term at the associated multiplier. Even so, at this multiplier, $q(x,\lambda)$ in \eqref{eq:q-nearest} differs from that value and does not satisfy the moment condition. In the counterexample to Theorems~5 and~6, the observed vector has a density with respect to Lebesgue measure, and the assumptions used in Theorems~5 and~6 hold. The first components of the multipliers under the OTGMM estimator and the GMM estimator with modified moments converge to different probability limits.

\subsection{Counterexample to Theorem~4}
\label{sec:theorem4-scalar}

Set $n=1$, $d_x=d_g=1$, and $d_\theta=0$. Problem \eqref{eq:primal} then has one scalar variable $z$ and one scalar moment condition, and there is no optimization over $\theta$. Let $\mathcal X=\R$, let $x=-3/4$, and define the scalar moment function
\begin{align}
g(z)=-z^4+z^3+z^2-1=-(z-1)(z^3-z-1).
\label{eq:scalar-moment}
\end{align}
The sample version of the problem becomes
\begin{align}
\min_z \frac{1}{2}\left(z+\frac{3}{4}\right)^2
\quad\text{subject to}\quad
g(z)=0.
\label{eq:scalar-primal}
\end{align}
Let $\rho=1.324717957244746\ldots$ denote the unique real solution of $z^3-z-1=0$. The real feasible values in \eqref{eq:scalar-primal} are $1$ and $\rho$. In this example, \eqref{eq:q-nearest} and \eqref{eq:modified-moments} reduce to $q(x,\lambda)$ and $\widetilde g(x,\lambda)=g(q(x,\lambda))$, and we write $\ell(z;x,\lambda)=\frac{1}{2}(z-x)^2-\lambda g(z)$.

\begin{proposition}[Counterexample to Theorem~4]\label{prop:theorem4-scalar}
The unique global minimizer of \eqref{eq:scalar-primal} is $z^{\text{prime}}=1$, with associated multiplier $\lambda^{\text{prime}}=7/4$. At this multiplier, $q(x,\lambda^{\text{prime}})$ is a strict local minimizer of $\ell$ and does not satisfy the moment condition.
\end{proposition}

\begin{proof}
Since the inequality $\abs{1+3/4}<\abs{\rho+3/4}$ holds, the unique feasible value with minimum transport cost is $z^{\text{prime}}=1$. Because the equality $g'(1)=1$ holds, the scalar form of \eqref{eq:foc-z} gives $\lambda^{\text{prime}}=(1-x)/g'(1)=7/4$. At $\lambda^{\text{prime}}$, the first-order condition \eqref{eq:foc-z} becomes $z-(7/4)g'(z)=-3/4$, which is equivalent to
\begin{align}
\frac{1}{4}(z-1)(28z^2+7z-3)=0.
\label{eq:stationary-factor}
\end{align}
The three solutions are $\alpha_-=(-7-\sqrt{385})/56$, $\alpha_+=(-7+\sqrt{385})/56$, and $1$. Their distances from $x=-3/4$ are $(35-\sqrt{385})/56$, $(35+\sqrt{385})/56$, and $7/4$, respectively. Equation~\eqref{eq:q-nearest} therefore gives $q(x,\lambda^{\text{prime}})=\alpha_-=-0.4753824441\ldots$. The value of the modified moment function at this multiplier is $\widetilde g(x,\lambda^{\text{prime}})=g(\alpha_-)=(-23493+5\sqrt{385})/25088<\nobreak0$. Thus, $\lambda^{\text{prime}}$ does not satisfy the sample moment equation $\widetilde g(x,\lambda)=0$. We also have
\begin{align}
1-\frac{7}{4}g''(\alpha_-)
=
\frac{55+9\sqrt{385}}{32}
>0.
\label{eq:alpha-hessian}
\end{align}
The inequality in \eqref{eq:alpha-hessian} shows that $\alpha_-$ is a strict local minimizer of $\ell$. A local second-order condition therefore does not imply that $q(x,\lambda^{\text{prime}})$ is the global minimizer of $\ell$.
\end{proof}

The global optimality of $z^{\text{prime}}$ can also be verified directly. For every $z\in\R$, the identity
\begin{align}
\frac{1}{2}\left(z+\frac{3}{4}\right)^2
-
\frac{7}{4}g(z)
-
\frac{49}{32}
=
\frac{1}{4}(z-1)^2(7z^2+7z+2)
\label{eq:global-identity}
\end{align}
holds. Since the inequality $7z^2+7z+2>0$ holds for every $z\in\R$, the function $\ell(z;x,\lambda^{\text{prime}})$ has the unique global minimizer $z=1$. This value is feasible, and the minimum of $\ell(z;x,\lambda^{\text{prime}})$ equals the minimum transport cost in \eqref{eq:scalar-primal}. The equivalence stated in Theorem~4 does not hold in this example because the relation $q(x,\lambda^{\text{prime}})=\alpha_-\neq z^{\text{prime}}$ holds. The pair $(z^{\text{prime}},\lambda^{\text{prime}})$ satisfies conditions~\eqref{eq:foc-lambda} and \eqref{eq:foc-z}, whereas $\lambda^{\text{prime}}$ does not satisfy the sample moment equation based on $\widetilde g$. The sample moment equations based on $\widetilde g$ are therefore not equivalent to the first-order conditions of \eqref{eq:primal} either.

Figure~\ref{fig:lagrangian} plots $\ell(z;x,\lambda^{\text{prime}})$.

\begin{figure}[t]
\centering
\includegraphics[width=0.88\textwidth]{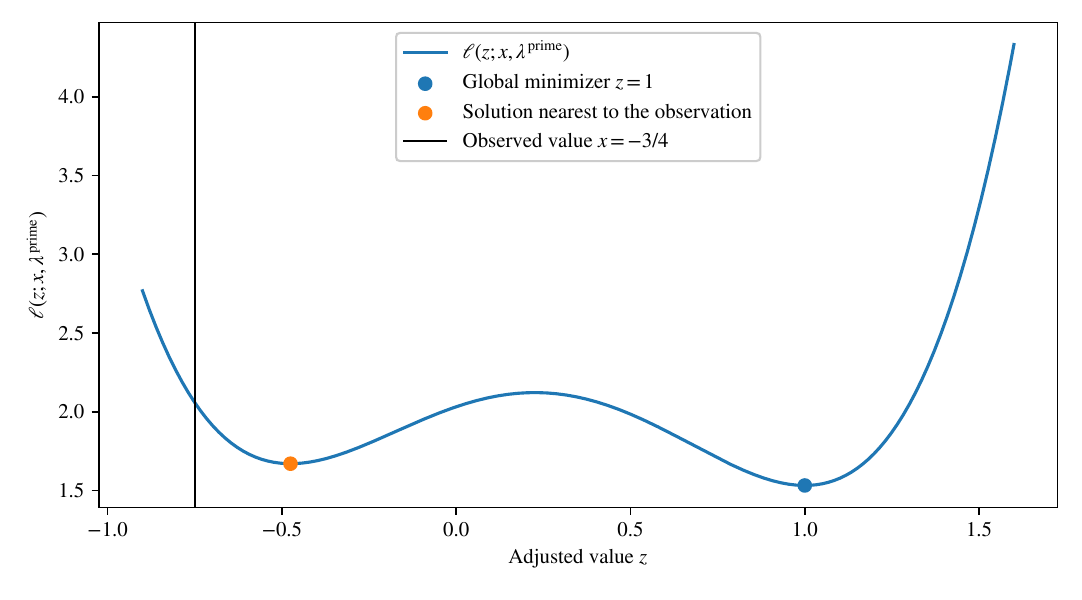}
\caption{The Lagrangian term in the counterexample to Theorem~4 in \citet{Schennach2026optimallytransported}. Among the values satisfying \eqref{eq:foc-z}, the one nearest to the observation is a strict local minimizer of $\ell$. The value $z=1$ is the unique global minimizer and the feasible value with the lowest transport cost.}
\label{fig:lagrangian}
\end{figure}

To exhibit a value of $\lambda$ satisfying $\widetilde g(x,\lambda)=0$, define
\begin{align}
\lambda^{\text{modified}}
=
\frac{\rho+3/4}{g'(\rho)}
=
-1.4982044980\ldots.
\label{eq:lambda-modified}
\end{align}
Appendix~\ref{app:modified-solution} shows that, at $\lambda^{\text{modified}}$, condition~\eqref{eq:foc-z} has $z=\rho$ as its only real solution, so we have $q(x,\lambda^{\text{modified}})=\rho$ and $\widetilde g(x,\lambda^{\text{modified}})=g(\rho)=0$. No other multiplier solves $\widetilde g(x,\lambda)=\nobreak0$. Any solution satisfies $q(x,\lambda)\in\{1,\rho\}$. The scalar form of \eqref{eq:foc-z} then gives $\lambda=7/4$ in the first case and $\lambda=\lambda^{\text{modified}}$ in the second. The first case is excluded because the equality $q(x,7/4)=\alpha_-$ holds. Using $\rho^3=\rho+1$, we have $\lambda^{\text{modified}}=(4\rho+3)/(4(3\rho^2-2\rho-4))<0$ and $1-\lambda^{\text{modified}}g''(\rho)=(12\rho^2+7\rho+13)/(2(3\rho^2-2\rho-4))<0$.
{\emergencystretch=1em
The inequalities follow from $1<\rho<3/2$. Because the leading quartic coefficient of $\ell(z;x,\lambda^{\text{modified}})$ is $\lambda^{\text{modified}}<0$, this function is unbounded below. Although the equality $\widetilde g(x,\lambda^{\text{modified}})=0$ holds, the value $q(x,\lambda^{\text{modified}})=\rho$ is a local maximum of $\ell$. The transport cost at $z^{\text{prime}}=1$ is $49/32=1.53125$, whereas the cost at $q(x,\lambda^{\text{modified}})=\rho$ is $\frac{1}{2}(\rho+3/4)^2=2.1522273011\ldots$, approximately $40.6$ percent above the global minimum.\par}

\begin{remark}[Uniqueness of the solution in $z$ to \eqref{eq:foc-z} does not imply global optimality]\label{rem:unique-inverse}
A solution of the population equations based on $\widetilde g$ can differ from the global minimizer of the original constrained problem even when \eqref{eq:foc-z} has a unique solution in $z$. Let $x$ be uniform on $[0,9/2]$, and consider $g_1(z,\theta)=\theta-3z+2$ and $g_2(z,\theta)=\theta-z^2$. Set $g(z,\theta)=(g_1(z,\theta),g_2(z,\theta))\trans$. The first-order condition of the sample Lagrangian \eqref{eq:sample-lagrangian} with respect to $\theta$ gives $\lambda_2=-\lambda_1$, so for $\lambda=(\lambda_1,-\lambda_1)\trans$, condition~\eqref{eq:foc-z} has the unique solution $q(x,\theta,\lambda)=(x-3\lambda_1)/(1-2\lambda_1)$ whenever the inequality $\lambda_1\neq1/2$ holds. Substituting this value into the population equations based on $\widetilde g$, and using $\E[x]=9/4$ and $\E[x^2]=27/4$, gives
\begin{align}
\E[(q(x,\theta,\lambda)-1)(q(x,\theta,\lambda)-2)]
=
-\frac{(\lambda_1-2)(\lambda_1+1)}{(2\lambda_1-1)^2}
=0.
\end{align}
This equation has the two roots $\lambda_1=-1$ and $\lambda_1=2$. The first gives $q(x,\theta,\lambda)=1+x/3$ and $\theta=13/4$, while the second gives $q(x,\theta,\lambda)=2-x/3$ and $\theta=7/4$.

The two constraints in the original problem imply $\E[(z-3/2)^2]=1/4$. In the Hilbert space $L^2$, the unique random variable satisfying these constraints and minimizing $\E[(z-x)^2]$ is
\begin{align}
z
=
\frac{3}{2}
+
\frac{1/2}{(\E[(x-3/2)^2])^{1/2}}
\left(x-\frac{3}{2}\right)
=
1+\frac{x}{3}.
\end{align}
The value $\lambda_1=-1$ gives the unique minimum-cost value of $z$ satisfying the two constraints, whereas $\lambda_1=2$ gives the value farthest from $x$ among those satisfying them on the same sphere. For each multiplier, condition~\eqref{eq:foc-z} has one solution in $z$, and $\ell$ is globally convex at $\lambda_1=-1$. At $\lambda_1=2$, however, the second derivative of $\ell$ is $1-2\lambda_1=-3$, so uniqueness of $z$ in \eqref{eq:foc-z} does not eliminate the root $\lambda_1=2$ of the population equations based on $\widetilde g$.
\end{remark}

Supplemental Appendix~S.2 gives a counterexample using the moment conditions from a linear regression with a constant and the regressor as instruments. In that example, the sample moment equations based on $\widetilde g$ have an additional solution.

The scalar example has $d_\theta=0$. Proposition~S.1 in Supplemental Appendix~S.1 gives an example with $d_\theta=1$ and $d_g=2$. In that example, the global minimizer of the original constrained problem has parameter component $\theta=-11/8$, whereas a solution of the sample moment equations based on $\widetilde g$ has parameter component $\theta=-11\rho^2/8=-2.4129567911\ldots$.

\subsection{Different Probability Limits for the First Component of the Lagrange Multiplier (Theorems~5 and~6)}
\label{sec:large-error}

We now construct a model with $d_\theta=1$ and $d_g=2$. For $\delta>0$, the first observed component $x_\delta$ is uniform on $[-3/4-\delta,-3/4+\delta]$. The first moment function is the same polynomial as in the scalar example, so the difference between the original constrained problem and the moment equation based on $\widetilde g$ persists for sufficiently small $\delta$. A second random variable $y$, independent of $x_\delta$, and a second moment condition determine $\theta$ and the second component of $\lambda$ and make the covariance matrix of the modified moments nonsingular. Let $\lambda_1^{\text{prime}}(\delta)$ denote the probability limit of the first component of the Lagrange multiplier under the OTGMM estimator, and let $\lambda_1^{\text{modified}}(\delta)$ denote the corresponding limit under the GMM estimator with modified moments. Appendix~\ref{app:continuous-proof} gives the model and the main steps of the proof, and Supplemental Appendix~S.3 completes the verification of the assumptions.

\begin{proposition}[Counterexample to Theorems~5 and~6]\label{prop:population}
There exists a model with one scalar parameter and two moment conditions. The observed vector has a density with respect to Lebesgue measure, and the model satisfies Assumptions~1, 2, 12, 13, 14, 18, 19, and~20 of \citet{Schennach2026optimallytransported}. In this model, as $n\to\infty$, the first components of the multipliers under the OTGMM estimator and the GMM estimator with modified moments converge to different probability limits. For each estimator, the parameter and multiplier estimates, centered at their probability limits and multiplied by $\sqrt n$, converge jointly in distribution to a Gaussian distribution.
\end{proposition}

This model also provides a counterexample to the equivalence stated in Theorem~4. The asymptotic analysis in Theorems~5 and~6 applies to the GMM estimator with modified moments. For asymptotic normality, this estimator is centered at the solution of $\E[\widetilde g(x,\widetilde\theta)]=0$. In this model, the first component of the Lagrange multiplier in that solution is not the probability limit of the corresponding component under the OTGMM estimator, so this analysis does not establish the stated joint asymptotic results for the OTGMM estimator.

For $\delta=0.01$, numerical quadrature gives $\lambda_1^{\text{prime}}(\delta)\approx1.7500165364$ and $\lambda_1^{\text{modified}}(\delta)\approx-1.4982087697$. In this model, the parameter $\theta$ and the second component of $\lambda$ have the same limits under both estimators. Supplemental Appendix~S.1 gives a one-observation model in which the values of $\theta$ also differ.

\section{The Population Value Selected by OTGMM}
\label{sec:parameter}

Sections~\ref{sec:small-error} and~\ref{sec:large-problems} show that the stated assumptions do not justify the asymptotic conclusions for the OTGMM estimator. We next consider the parameter value selected by the population optimization problem in \eqref{eq:population-problem}. If the observed distribution satisfies the moment conditions at some parameter value, the minimum transport cost at that value is zero. Under misspecification, the minimizing value depends on the transport metric and on which variables the estimator may adjust.

\subsection{Zero Transport Cost and OLS Attenuation}

We first show that, with only the OLS moment condition in a classical errors-in-variables model, OTGMM does not remove classical attenuation and returns the usual OLS coefficient. We establish this result by showing that OTGMM leaves the data unchanged whenever the moment conditions hold at the observed data.

\begin{proposition}[Zero transport cost]\label{prop:zero-cost}
For any $\theta\in\Theta$, there exists a probability measure $\mu_{zx}$ satisfying the constraints in \eqref{eq:population-problem} with zero transport cost if and only if
\begin{align}
\E_{\mu_x}[g(x,\theta)]=0.
\label{eq:observed-moment-equation}
\end{align}
At the sample level, if $n^{-1}\sum_i g(x_i,\theta)=0$ has the unique solution $\theta_n^\dagger$, then the parameter component $\widehat\theta_n$ of every global minimizer of \eqref{eq:primal} satisfies $\widehat\theta_n=\theta_n^\dagger$, and the corresponding values of $z_i$ satisfy $z_i=x_i$ for every $i$.
\end{proposition}

\begin{proof}
If \eqref{eq:observed-moment-equation} holds, setting $z=x$ gives a probability measure $\mu_{zx}$ satisfying the constraints in \eqref{eq:population-problem} with zero cost. Conversely, zero transport cost implies $z=x$ almost surely, so the population moment condition gives \eqref{eq:observed-moment-equation}. The sample statement follows in the same way. The values $z_i=x_i$ and $\theta=\theta_n^\dagger$ attain objective value zero, so every global minimizer of \eqref{eq:primal} also satisfies $z_i=x_i$ for every $i$.
\end{proof}

Let $\Theta=\R$, let $\beta_0\in\R$ denote the coefficient in the latent-variable equation, and consider the centered errors-in-variables model
\begin{align}
y=\beta_0 z+\varepsilon_y,
\qquad
x=z+u,
\label{eq:measurement-error-model}
\end{align}
where $z$ is the latent regressor, $u$ is the measurement error in the regressor, and $\varepsilon_y$ is the error term. These three variables are mutually independent, have mean zero, have finite second moments, and satisfy $\bbV(z)+\bbV(u)>0$. For $\beta\in\Theta$, the OLS moment condition is
\begin{align}
g((y,x),\beta)=x(y-\beta x).
\label{eq:ols-moment}
\end{align}
For a sample $(y_i,x_i)_{i=1}^n$ with $\sum_i x_i^2>0$, the sample moment equation has the unique solution
\begin{align}
\widehat\beta^\dagger
=
\frac{\sum_{i=1}^n x_i y_i}{\sum_{i=1}^n x_i^2}.
\label{eq:sample-ols}
\end{align}
Proposition~\ref{prop:zero-cost} implies that every global minimizer of \eqref{eq:primal} for the moment condition in \eqref{eq:ols-moment} leaves the data unchanged and has parameter component $\widehat\beta^\dagger$. By the law of large numbers, we have
\begin{align}
\widehat\beta^\dagger
\xrightarrow{p}
\beta_0
\frac{\bbV(z)}
{\bbV(z)+\bbV(u)}.
\label{eq:attenuation}
\end{align}
Thus, optimal transport alone does not remove classical attenuation. The example is just identified and therefore lacks the overidentification for which OTGMM was proposed. Overidentifying restrictions can make zero transport cost impossible, but additional moment conditions and assumptions are needed to equate the distribution closest to the observed distribution under the chosen cost with the true latent distribution. These conditions must link the observed variables, the latent variables, and the measurement error.

\subsection{Dependence on the Transport Metric}

We next show that, under misspecification, changing the transport metric or the variables that may be adjusted can change the population value selected by OTGMM. Let $\Theta=\R$ and $\mathcal X=\R^2$, write $x=(x_1,x_2)\trans$ and $z=(z_1,z_2)\trans$, and let $\mu_j=\E_{\mu_x}[x_j]$ for $j=1,2$. Consider
\begin{align}
g(z,\theta)
=
\begin{pmatrix}
z_1-\theta\\
z_2-\theta
\end{pmatrix}.
\label{eq:location-moments}
\end{align}
The objective in \eqref{eq:population-problem} uses $\norm{z-x}^2=(z_1-x_1)^2+(z_2-x_2)^2$. The original study notes that a weighted Euclidean norm may be used instead to represent different expected error magnitudes across coordinates \citep[p.~620]{Schennach2026optimallytransported}. We keep the feasible set in \eqref{eq:population-problem} and replace only the squared distance by $a_1(z_1-x_1)^2+a_2(z_2-x_2)^2$, where $a_1,a_2>0$. The objective becomes
\begin{align}
\frac{1}{2}\E_{\mu_{zx}}\left[
a_1(z_1-x_1)^2+a_2(z_2-x_2)^2
\right],
\qquad a_1,a_2>0.
\label{eq:weighted-transport-cost}
\end{align}
For each $\theta$, Jensen's inequality gives the lower bound $\frac{1}{2}(a_1(\theta-\mu_1)^2+a_2(\theta-\mu_2)^2)$, and, for $j=1,2$, the constant shifts $z_j=x_j+\theta-\mu_j$ attain this bound. The lower bound is minimized at $\theta=(a_1\mu_1+a_2\mu_2)/(a_1+a_2)$. If the two means differ, changing the relative costs changes the selected value. Under the unweighted Euclidean norm, rescaling one coordinate has the same effect, because it changes that coordinate's weight relative to the other. The selected value also depends on which variables may be adjusted. Imposing $z_2=x_2$ forces $\theta=\mu_2$, whereas imposing $z_1=x_1$ forces $\theta=\mu_1$.

Under misspecification, GMM and GEL also select pseudo-true values that depend on their objective functions. OTGMM interprets $z$ as latent data and $x-z$ as error. Suppose that a separate latent-variable model defines a parameter $\theta_0$. For $Q$ in \eqref{eq:population-problem}, a population separation condition used for consistency at $\theta_0$ is, for every $\epsilon>0$,
\begin{align}
Q(\theta_0)
<
\inf_{\norm{\theta-\theta_0}\geq\epsilon}Q(\theta).
\label{eq:population-separation}
\end{align}
That $x$ is an error-contaminated version of some $z$ satisfying the moment conditions does not imply \eqref{eq:population-separation}. Interpreting $z$ as latent data requires that the joint distribution of the latent and observed variables minimize the chosen cost among all distributions satisfying the moment conditions.

Assumption~12 does not provide this connection because it defines $(\theta_0,\lambda_0)$ as the unique solution of $\E[\widetilde g(x,\theta,\lambda)]=0$. Assumption~13 requires $Q(\theta)$ to be uniquely minimized at $\theta_0$ but does not imply that this minimizer equals the parameter defined by an independently specified measurement-error model. Equating the two requires an economic or statistical reason why $Q$ is minimized at the parameter of that model.

\subsection{Relation to Distributional Synthetic Controls}

Distributional synthetic controls illustrate the same distinction between changing an estimator and changing the model that defines the parameter being estimated. \citet{Gunsilius2023distributionalsynthetic} matches entire quantile functions in Wasserstein space and shows that the resulting weights generally differ from those obtained by matching means, whereas \citet{Kato2025asymptoticallyunbiased} addresses attenuation by replacing a linear relation in expected outcomes with a mixture model for outcome distributions. In both studies, the distributional objective or the maintained model changes the weights being estimated, so describing such a method as a correction to the original mean model requires a separate argument linking the parameters in the two models.

\section{Corrected Results}
\label{sec:corrected}

We give separate sufficient conditions for the small-error and large-error analyses. For the small-error results, we impose a bound on how much $g(z,\theta)$ can change when $z$ is close to $x$, uniformly over $\theta$. For the large-error analysis, a sufficient condition for minimum-cost adjustments at a given parameter value is that the values $q(x_i,\theta,\lambda)$ used in \eqref{eq:modified-moments} minimize each Lagrangian term and jointly satisfy the sample moment condition. The resulting transport cost must then be minimized over $\theta$ to obtain a global OTGMM minimizer.

\subsection{Consistency of Parameter Components of Global Minimizers (Theorem~2)}

The condition below rules out the behavior in Proposition~\ref{prop:small-error}, where $g(z,\theta)-g(x,\theta)$ becomes arbitrarily large for $z$ close to $x$.

\begin{theorem}[Consistency under a uniform bound on changes in the moments (Theorem~2 in \citet{Schennach2026optimallytransported})]\label{thm:corrected-consistency}
Suppose that $x_i$ are i.i.d., $\Theta$ is compact, $\E[g(x,\theta)]$ is continuous in $\theta$, and $\theta_0$ is the unique solution of $\E[g(x,\theta)]=0$. Suppose that we have
\begin{align}
\sup_{\theta\in\Theta}
\norm{\widehat{\E}[g(x,\theta)]-\E[g(x,\theta)]}
\xrightarrow{p}0.
\label{eq:corrected-ulln}
\end{align}
Assume that there are a nonnegative measurable function $L$ with $\E[L(x)^2]<\infty$ and a constant $C\geq0$ such that
\begin{align}
\norm{g(z,\theta)-g(x,\theta)}
\leq
L(x)\norm{z-x}+C\norm{z-x}^2
\label{eq:uniform-transport-lipschitz}
\end{align}
for every $x,z\in\mathcal X$ and every $\theta\in\Theta$. Suppose that the sample version of the problem has a global minimizer and that $\widehat Q(\theta_0)=o_p(1)$. Then, as $n\to\infty$, we have
\begin{align}
\sup_{\theta\in\widehat\Theta_n}\norm{\theta-\theta_0}
\xrightarrow{p}0.
\label{eq:corrected-global-consistency}
\end{align}
\end{theorem}

Condition~\eqref{eq:uniform-transport-lipschitz} adds a quadratic term to the Lipschitz form used in Definition~1 of \citet{Schennach2026optimallytransported} and uses a square-integrable envelope independent of $\theta$. Definition~1 permits a $\theta$-dependent envelope subject to the dominance condition in equation~(10) of that study. The quadratic term allows moment functions that are quadratic in $z$, such as linear instrumental-variable moments. When their coefficients are continuous in $\theta$ on the compact set $\Theta$ and $\E[\norm{x}^2]<\infty$, condition~\eqref{eq:uniform-transport-lipschitz} holds with $L(x)=c(1+\norm{x})$ for a constant $c$. In the large-error analysis, the original study proposes Assumptions~15--17 as an alternative to Assumption~14, and Assumption~17 uses this definition. The assumptions of the small-error analysis do not impose a condition of this form. The proof of Theorem~2 requires that, whenever $\norm{\widehat{\E}[g(x,\theta)]}$ is bounded away from zero, the moment condition cannot be satisfied by values $z_1,\ldots,z_n$ whose average squared distance from $x_1,\ldots,x_n$ converges to zero as $n\to\infty$.

\subsection{Asymptotic Normality of Parameter Components of Global Minimizers (Theorem~3)}

The next theorem assumes that the parameter components of all global minimizers lie within $O_p(n^{-1/2})$ of $\theta_0$. It also requires a local expansion of the sample moments and a quadratic approximation to the minimized transport cost that hold uniformly throughout that neighborhood.

\begin{theorem}[Asymptotic normality (Theorem~3 in \citet{Schennach2026optimallytransported})]\label{thm:corrected-normality}
Suppose that $\theta_0$ is an interior point of $\Theta$, that $\widehat\Theta_n$ is nonempty, and that we have
\begin{align}
\sup_{\theta\in\widehat\Theta_n}
\sqrt n\norm{\theta-\theta_0}=O_p(1).
\label{eq:n-half-global-rate}
\end{align}
Suppose that, as $n\to\infty$, we have $\sqrt n\,\widehat{\E}[g(x,\theta_0)]\xrightarrow{d}\mathcal N(0,\Omega_g)$ for a positive-definite matrix $\Omega_g\in\R^{d_g\times d_g}$. Suppose that there are a full-column-rank matrix $G\in\R^{d_g\times d_\theta}$, a positive-definite matrix $M_Q\in\R^{d_g\times d_g}$, and random positive-definite matrices $\widehat M_{Q,n}\in\R^{d_g\times d_g}$ such that $\widehat M_{Q,n}\xrightarrow{p}M_Q$. Suppose also that, for every $C_h>0$, the following two relations hold:
\begin{align}
\sup_{\norm{h}\leq C_h}
\norm{
\sqrt n
\left(
\widehat{\E}[g(x,\theta_0+h/\sqrt n)]-\widehat{\E}[g(x,\theta_0)]
\right)-Gh
}
=o_p(1)
\label{eq:uniform-local-moment-expansion}
\end{align}
and
\begin{align}
\sup_{\norm{h}\leq C_h}
\left|
\begin{aligned}
&n\widehat Q(\theta_0+h/\sqrt n)\\
&-
\frac{n}{2}
\widehat{\E}[g(x,\theta_0+h/\sqrt n)]\trans
\widehat M_{Q,n}^{-1}
\widehat{\E}[g(x,\theta_0+h/\sqrt n)]
\end{aligned}
\right|
=o_p(1).
\label{eq:uniform-local-cost-expansion}
\end{align}
Let $\widehat\theta_n$ be any measurable function of the sample satisfying $\widehat\theta_n\in\widehat\Theta_n$. Then, the following expansion holds:
\begin{align}
\sqrt n(\widehat\theta_n-\theta_0)
=
-
(G\trans M_Q^{-1}G)^{-1}G\trans M_Q^{-1}
\sqrt n\,\widehat{\E}[g(x,\theta_0)]
+o_p(1).
\label{eq:corrected-linear-representation}
\end{align}
Define $V$ by
\begin{align}
V=
(G\trans M_Q^{-1}G)^{-1}
G\trans M_Q^{-1}\Omega_g M_Q^{-1}G
(G\trans M_Q^{-1}G)^{-1}.
\label{eq:corrected-variance}
\end{align}
Consequently, as $n\to\infty$, we have
\begin{align}
\sqrt n(\widehat\theta_n-\theta_0)
\xrightarrow{d}
\mathcal N(0,V).
\end{align}
\end{theorem}

Our Theorem~\ref{thm:corrected-normality} assumes rather than derives the $O_p(n^{-1/2})$ bound in \eqref{eq:n-half-global-rate} and the two uniform approximations in \eqref{eq:uniform-local-moment-expansion} and \eqref{eq:uniform-local-cost-expansion}. Our Theorem~\ref{thm:corrected-consistency} gives consistency but not, by itself, the bound in \eqref{eq:n-half-global-rate}. One way to obtain that bound is to construct feasible values $z_1,\ldots,z_n$ at $\theta_0$ with cost $O_p(n^{-1})$ and prove a local quadratic lower bound for $\widehat Q(\theta)$. Smoothness and rank conditions must then be shown to imply that \eqref{eq:uniform-local-cost-expansion} holds uniformly over $\theta_0+h/\sqrt n$ with $\norm{h}\leq C_h$ for every $C_h>0$. A local expansion around one solution in $z$ to \eqref{eq:foc-z} does not establish this uniform approximation.

\subsection[Minimization over z at a Given Parameter Value (Theorem~4)]{Minimization over $z_i$ at a Given Parameter Value (Theorem~4)}

Condition~\eqref{eq:foc-z} does not imply that $q(x_i,\theta,\lambda)$ minimizes the corresponding term of the Lagrangian. The next proposition adds this minimization requirement and the sample moment condition.

\begin{proposition}[Minimization over $z_i$ at a given parameter value (Theorem~4 in \citet{Schennach2026optimallytransported})]\label{prop:repair}
Let $\theta\in\Theta$ be given. Suppose that there is a multiplier $\lambda$ for which $q(x_i,\theta,\lambda)$ in \eqref{eq:q-nearest} is single-valued for every $i$ and
\begin{align}
q(x_i,\theta,\lambda)\in\argmin_{z\in\mathcal X}
\left(
\frac{1}{2}\norm{z-x_i}^2-\lambda\trans g(z,\theta)
\right),
\label{eq:global-q}
\end{align}
and suppose that
\begin{align}
\widehat{\E}[g(q(x,\theta,\lambda),\theta)]=0.
\label{eq:global-q-feasible}
\end{align}
Then, the collection $(q(x_i,\theta,\lambda))_{i=1}^n$ minimizes the transport cost over $z_1,\ldots,z_n$ subject to the sample moment condition at this value of $\theta$. If the minimizer in \eqref{eq:global-q} is unique for every $i$, then this minimizing collection is unique.
\end{proposition}

In \eqref{eq:primal}, minimizing over $z_1,\ldots,z_n$ for a fixed $\theta$ is the inner minimization, and its minimizers depend on the full sample rather than on $x_i$ alone. In a representation $z_i=q(x_i,\theta,\lambda)$, the common multiplier $\lambda$ can depend on the full sample.

Proposition~\ref{prop:repair} concerns minimization only over $z_1,\ldots,z_n$. To obtain a global minimizer of \eqref{eq:primal}, one must also minimize the resulting transport cost over $\theta$. A stationary point in $\theta$ need not be a global minimizer without an additional condition such as convexity.

For the model with affine moment functions in the following corollary, the conditions in Proposition~\ref{prop:repair} hold and the minimizing values of $z_i$ have a closed form.

\begin{corollary}[Affine moment functions at a given parameter value (Theorem~4 in \citet{Schennach2026optimallytransported})]\label{cor:affine}
Let $\theta\in\Theta$ be given, let $\mathcal X=\R^{d_x}$, and suppose that $g(z,\theta)=A_{\mathrm{aff}}(\theta)z-b_{\mathrm{aff}}(\theta)$, where $A_{\mathrm{aff}}(\theta)\in\R^{d_g\times d_x}$, $b_{\mathrm{aff}}(\theta)\in\R^{d_g}$, and $A_{\mathrm{aff}}(\theta)$ has full row rank. Define
\begin{align}
\lambda=
\left(A_{\mathrm{aff}}(\theta)A_{\mathrm{aff}}(\theta)\trans\right)^{-1}
\left(b_{\mathrm{aff}}(\theta)-A_{\mathrm{aff}}(\theta)\widehat{\E}[x]\right).
\label{eq:affine-solution}
\end{align}
Then, for every observation, condition~\eqref{eq:foc-z} has the unique solution $q(x_i,\theta,\lambda)=x_i+A_{\mathrm{aff}}(\theta)\trans\lambda$. These values satisfy the sample moment condition and uniquely minimize the transport cost over $z_1,\ldots,z_n$ at this value of $\theta$.
\end{corollary}

For nonlinear moment functions with $\mathcal X=\R^{d_x}$, let $g_k$ and $\lambda_k$ denote the $k$th components of $g$ and $\lambda$, and let $\partial_{zz\trans}g_k(z,\theta)$ denote the Hessian of $g_k$ with respect to $z$. Suppose that $g_k(\cdot,\theta)$ is twice continuously differentiable with respect to $z$ for every $\theta\in\Theta$ and $k=1,\ldots,d_g$. Let $I$ denote the $d_x\times d_x$ identity matrix, let $\succeq$ mean that the matrix on the left minus the matrix on the right is positive semidefinite, and let $\Lambda\subset\R^{d_g}$ denote the set of multiplier values under consideration. Suppose that there is an $\varepsilon_H>0$ such that
\begin{align}
I-
\sum_{k=1}^{d_g}\lambda_k\partial_{zz\trans}g_k(z,\theta)
\succeq
\varepsilon_H I
\label{eq:weighted-hessian-sufficient}
\end{align}
for every $z\in\R^{d_x}$, $\theta\in\Theta$, and $\lambda\in\Lambda$. Under this condition, the Hessian of $\ell$ with respect to $z$ is uniformly positive definite, so $\ell$ is strongly convex and coercive. Condition~\eqref{eq:foc-z} therefore has a unique solution in $z$, the unique global minimizer of $\ell$. Inequality~\eqref{eq:weighted-hessian-sufficient} also bounds the operator norm of the inverse matrix in the implicit function theorem. Strong convexity is sufficient but not necessary for \eqref{eq:global-q}; Proposition~\ref{prop:repair} also requires \eqref{eq:global-q-feasible}. For the weighted cost in \eqref{eq:weighted-transport-cost}, the corresponding Hessian condition is obtained by replacing $I$ in \eqref{eq:weighted-hessian-sufficient} with $\operatorname{diag}(a_1,a_2)$. If the condition holds only locally, additional assumptions must ensure that a minimizer exists, is unique, and remains in the region where the condition holds. Section~\ref{sec:assump16} explains why Assumption~16 does not imply \eqref{eq:weighted-hessian-sufficient}.

\subsection{Consistency and Asymptotic Normality for the GMM Estimator with Modified Moments (Theorems~5 and~6)}

Theorems~5 and~6 in \citet{Schennach2026optimallytransported} analyze the asymptotic behavior of the estimator based on the modified moment function in \eqref{eq:modified-moments}. The following proposition states sufficient conditions for consistency and asymptotic normality of this estimator.

\begin{proposition}[GMM estimator with modified moments (Theorems~5 and~6 in \citet{Schennach2026optimallytransported})]\label{prop:modified-gmm}
Let $\Lambda\subset\R^{d_g}$ denote the set of multiplier values, and suppose that $\Theta\times\Lambda$ is compact. Suppose that there are a set $N_q\subset\mathcal X$ with $\mu_x(N_q)=0$ and a single-valued function $q:\mathcal X\times\Theta\times\Lambda\to\mathcal X$, measurable in $(x,\theta,\lambda)$, such that $q(x,\theta,\lambda)$ belongs to the argmin set in \eqref{eq:q-nearest} for every $x\notin N_q$ and every $(\theta,\lambda)\in\Theta\times\Lambda$. Suppose also that $\E\sqb{\widetilde g(x,\widetilde\theta)}$ is continuous in $\widetilde\theta$ and that $\E\sqb{\widetilde g(x,\widetilde\theta)}=0$ has the unique solution $\widetilde\theta_0=\p{\theta_0\trans,\lambda_0\trans}\trans$. Assume that, as $n\to\infty$, we have
\begin{align}
\sup_{\widetilde\theta\in\Theta\times\Lambda}
\left\lVert
\widehat{\E}[\widetilde g(x,\widetilde\theta)]-\E[\widetilde g(x,\widetilde\theta)]
\right\rVert
\xrightarrow{p}0,
\label{eq:modified-moment-ulln}
\end{align}
and that the sample equations $\widehat{\E}[\widetilde g(x,\widetilde\theta)]=0$ have a solution in $\Theta\times\Lambda$ with probability approaching one. Then, every such solution is consistent for $\widetilde\theta_0$.

For asymptotic normality, suppose in addition that $\widetilde\theta_0$ is an interior point of $\Theta\times\Lambda$ and that $\widetilde g(x,\widetilde\theta)$ is continuously differentiable in $\widetilde\theta$ on a neighborhood of $\widetilde\theta_0$ almost surely. Suppose that there is a matrix $\Omega\in\R^{(d_\theta+d_g)\times(d_\theta+d_g)}$ such that, as $n\to\infty$, we have $\sqrt n\,\widehat{\E}[\widetilde g(x,\widetilde\theta_0)]\xrightarrow{d}\mathcal N(0,\Omega)$. Let $\mathcal U\subset\Theta\times\Lambda$ be a convex neighborhood of $\widetilde\theta_0$. Suppose that $\E[\partial_{\widetilde\theta\trans}\widetilde g(x,\widetilde\theta)]$ exists for every $\widetilde\theta\in\mathcal U$ and is continuous at $\widetilde\theta_0$. Define $\widetilde G=\E[\partial_{\widetilde\theta\trans}\widetilde g(x,\widetilde\theta_0)]$, and suppose that $\widetilde G\in\R^{(d_\theta+d_g)\times(d_\theta+d_g)}$ is nonsingular. Suppose also that we have
\begin{align}
\sup_{\widetilde\theta\in\mathcal U}
\left\lVert
\widehat{\E}[\partial_{\widetilde\theta\trans}\widetilde g(x,\widetilde\theta)]
-
\E[\partial_{\widetilde\theta\trans}\widetilde g(x,\widetilde\theta)]
\right\rVert
\xrightarrow{p}0.
\label{eq:modified-jacobian-ulln}
\end{align}
Under these additional assumptions, let $\widehat{\widetilde\theta}=\p{\widehat\theta\trans,\widehat\lambda\trans}\trans$ be any consistent solution of the sample moment equations based on $\widetilde g$. Then, as $n\to\infty$, we have
\begin{align}
\sqrt n\p{\widehat{\widetilde\theta}-\widetilde\theta_0}
\xrightarrow{d}
\mathcal N\p{0,\widetilde G^{-1}\Omega(\widetilde G^{-1})\trans}.
\label{eq:modified-gmm-limit}
\end{align}
\end{proposition}

Proposition~\ref{prop:modified-gmm} concerns the GMM estimator with modified moments for the function $q$ used in \eqref{eq:modified-moments}. Its conclusions also apply to an OTGMM estimator whose parameter component and associated multiplier solve the sample moment equations based on $\widetilde g$. Conditions~\eqref{eq:global-q} and \eqref{eq:global-q-feasible}, together with global minimization of the resulting transport cost over $\theta$, are sufficient for a solution of these equations to yield a global minimizer of \eqref{eq:primal}.

\section{Other Corrections and Numerical Calculations}
\label{sec:additional}

\subsection{Existence of a Population Minimizer (Problem~(1))}
\label{sec:population-existence}

The discussion following the population optimization problem in \citet{Schennach2026optimallytransported} states that the minimization problem has a solution whenever at least one distribution satisfies the moment condition. Such a distribution makes the constraint set nonempty, but the population objective need not attain its infimum. In the following example, the observed variable $x$ and each variable $z_b$ in the feasible sequence have densities with respect to Lebesgue measure.

Let $\epsilon\in(0,1/4]$, and let $x$ and $\zeta$ be independent and uniform on $[-\epsilon,\epsilon]$. Define the scalar moment function $g:\R\to\R$ by $g(z)=(z-1)e^z$. Set $m_0=\E[g(x)]<0$. For $b>1+\epsilon$, let $m_b=\E[g(b+\zeta)]>0$, and define $p_b=-m_0/(m_b-m_0)$. Let $\xi_b$ be a Bernoulli random variable with success probability $p_b$, independent of $(x,\zeta)$, and define $z_b=(1-\xi_b)x+\xi_b(b+\zeta)$. The distribution of $z_b$ is a mixture of two uniform distributions and therefore has a density. By the definition of $p_b$, we have $\E[g(z_b)]=(1-p_b)m_0+p_b m_b=0$, so the joint distribution of $(z_b,x)$ satisfies the constraints in \eqref{eq:population-problem}. As $b\to\infty$, its cost satisfies $\E\sqb{(z_b-x)^2}=p_b\E\sqb{(b+\zeta-x)^2}=p_b(b^2+2\epsilon^2/3)\to0$. The inequality $m_b\geq(b-1-\epsilon)e^{b-\epsilon}$ implies $p_b b^2\to0$ as $b\to\infty$. A joint distribution with zero cost would satisfy $z=x$ almost surely and hence $\E[g(z)]=m_0<0$, violating the constraints in \eqref{eq:population-problem}. The infimum is therefore zero, but no feasible joint distribution attains it.

\begin{proposition}[Existence of a population minimizer]\label{prop:population-existence}
Let $\theta\in\Theta$. Suppose that $\mathcal X$ is a closed subset of $\R^{d_x}$, the observed distribution $\mu_x$ has a finite second moment, and $g(\cdot,\theta)$ is bounded and continuous. If at least one probability measure $\mu_{zx}$ supported on $\mathcal X\times\mathcal X$, with marginal $\mu_x$ for $x$, satisfies the population moment condition and has finite transport cost, then the population optimization problem has a minimizer.
\end{proposition}

\subsection{Uniqueness of the Solution (Assumptions~12 and~13)}
\label{sec:assumption12}

The discussion following Assumption~12 in \citet{Schennach2026optimallytransported} states that Assumption~13 implies the uniqueness condition in Assumption~12. Uniqueness of the population minimizer and the transport map does not imply uniqueness of the multiplier when the moment conditions are linearly dependent.

{\emergencystretch=1em
\begin{proposition}[Assumption~13 does not imply Assumption~12]\label{prop:assumption12}
Let $\Theta=[-1/2,1/2]$, let $\Lambda=[-1,1]^2$, and let $x$ be uniform on $[-1,1]$. Define
\begin{align}
g(z,\theta)
=
\begin{pmatrix}
z-\theta\\
2(z-\theta)
\end{pmatrix}.
\label{eq:redundant-moments}
\end{align}
Then, Assumption~13 holds with $\theta_0=0$, but the population equation in Assumption~12 has infinitely many solutions in $\Theta\times\Lambda$.
\end{proposition}
}

Because the moment condition in the second component of \eqref{eq:redundant-moments} is a multiple of the first, uniqueness of the multiplier requires an additional rank condition.

\subsection{The Matrix Bound Used in the Supplement to the Original Study (Assumption~16)}
\label{sec:assump16}

We show that Assumption~16 implies neither the continuity conclusion of Lemma~7 in \citet{Schennach2026optimallytransported} nor the bound on the inverse matrix used in the proofs in its supplement. We give sufficient conditions for the inverse of the Lagrangian Hessian to exist and have bounded operator norm. Let $\operatorname{eigval}(B)$ denote the set of eigenvalues of a symmetric matrix $B$. Assumption~16 defines $\overline\lambda=\max_{\lambda\in\Lambda}\norm{\lambda}$ and $\overline\nu=\sup_{\theta\in\Theta}\allowbreak\sup_{z\in\mathcal X}\allowbreak\max_{1\leq k\leq d_g}\allowbreak\max\operatorname{eigval}(\partial_{zz\trans}g_k(z,\theta))$, and requires $\overline\lambda\,\overline\nu<1$. The supplement uses this condition to bound the operator norm of the inverse of
\begin{align}
I-
\partial_{zz\trans}(\lambda\trans g(z,\theta))
=
I-
\sum_{k=1}^{d_g}\lambda_k\partial_{zz\trans}g_k(z,\theta)
\label{eq:weighted-hessian}
\end{align}
by $(1-\overline\lambda\,\overline\nu)^{-1}$ \citep[pp.~4--5]{Schennach2026optimallytransportedSupp}.

The matrix in \eqref{eq:weighted-hessian} can fail to be positive definite even when Assumption~16 holds and all components of the multiplier are positive. In one dimension, let $g_1(z)=g_2(z)=(3/8)z^2$ and $\lambda=(1/\sqrt{2},1/\sqrt{2})\trans$. Set $g(z)=(g_1(z),g_2(z))\trans$ and $\Lambda=\{\lambda\}$. Since the Hessian of each moment function equals $3/4$, we have $\overline\lambda=1$ and $\overline\nu=3/4$, so Assumption~16 holds. However, we have $\sum_{k=1}^2\lambda_k g_k''(z)=3/(2\sqrt{2})>1$, so the Hessian of $z^2/2-\lambda\trans g(z)$ is negative. The bounds on the Hessians of the individual moment functions in Assumption~16 do not ensure positive definiteness of the matrix in \eqref{eq:weighted-hessian}.

A scalar example shows that the conclusion of Lemma~7 does not follow from Assumption~16. Let the moment function $g:\R\to\R$ and the observed value $x_0$ be given by $g(z)=-z^4/4$ and $x_0=2/(3\sqrt{3})$, and let $\Lambda$ contain an interval around $-1$. Because $g''(z)=-3z^2\leq0$ holds for every $z$, the definition in Assumption~16 gives $\overline\nu=0$, so the assumption holds for every compact $\Lambda$. The first-order condition \eqref{eq:foc-z} then becomes
\begin{align}
x_0=z+\lambda z^3.
\label{eq:fold-equation}
\end{align}
At $\lambda=-1$, equation~\eqref{eq:fold-equation} can be written as $-(z-1/\sqrt{3})^2(z+2/\sqrt{3})=0$, so $z=1/\sqrt{3}$ is a root of multiplicity two and the remaining root is $z=-2/\sqrt{3}$. The double root lies at distance $1/(3\sqrt{3})$ from $x_0$ and the other root at distance $8/(3\sqrt{3})$, so \eqref{eq:q-nearest} gives $q(x_0,-1)=1/\sqrt{3}$. For $\lambda<-1$, the maximum of $z+\lambda z^3$ on the positive half-line is below $x_0$, so the equation has only one real solution, which is negative and converges to $-2/\sqrt{3}$ as $\lambda\uparrow-1$. Hence, $q(x_0,\lambda)$ does not converge to $q(x_0,-1)$ as $\lambda\uparrow-1$. The function $q(x_0,\lambda)$ defined by \eqref{eq:q-nearest} is thus discontinuous at $\lambda=-1$ even though Assumption~16 holds. With $h(z)=z$ in the notation of Lemma~7, this discontinuity shows that the conclusion of Lemma~7 does not hold in this example.

The construction in Section~\ref{sec:large-error} also shows that Assumption~13(ii) does not imply Assumption~16 as stated. That construction satisfies Assumption~13, and its first moment function is the polynomial in \eqref{eq:scalar-moment}, whose second derivative has supremum $11/4$. Every multiplier set containing $(\lambda_1^{\text{prime}}(\delta),-1/2)$ has $\overline\lambda\geq\lambda_1^{\text{prime}}(\delta)$, which is close to $7/4$, so we have $\overline\lambda\,\overline\nu>1$, although each component of the optimal transport map is a smooth increasing function with a smooth inverse on the supports of the corresponding components of $x$ and $z$. Positive definiteness of the Lagrangian Hessian on those supports is weaker than the global componentwise condition in Assumption~16.

A direct sufficient condition is that there exists an $\varepsilon_H>0$ such that inequality~\eqref{eq:weighted-hessian-sufficient} holds for every $z\in\mathcal X$, $\theta\in\Theta$, and $\lambda\in\Lambda$. Under this condition, the inverse exists and its operator norm is at most $\varepsilon_H^{-1}$. With $\lVert\cdot\rVert_{\mathrm{op}}$ denoting the operator norm, a stronger condition that is sometimes easier to verify is
\begin{align}
\sup_{z\in\mathcal X,\,\theta\in\Theta,\,\lambda\in\Lambda}
\left\lVert
\sum_{k=1}^{d_g}
\lambda_k\partial_{zz\trans}g_k(z,\theta)
\right\rVert_{\mathrm{op}}
<1.
\label{eq:operator-hessian}
\end{align}
Either condition bounds the operator norm of the inverse of the matrix in \eqref{eq:weighted-hessian}. Under Assumption~16 as written, the conclusion of Lemma~7 does not follow, and the estimate used in the proof of Lemma~8 in \citet{Schennach2026optimallytransported} is not established.

\subsection{The Proposed Test of the Absence of Error (Theorem~6 and Equation~(14))}
\label{sec:absence-error}

After Theorem~6, \citet{Schennach2026optimallytransported} proposes a formal test of the absence of error and states the null hypothesis as $\lambda=0$. We show that Theorem~6 does not justify the proposed null distribution and derive a test with asymptotic null distribution $\chi^2_{d_g-d_\theta}$ under sufficient conditions. Let $\widehat\theta$ and $\widehat\lambda$ denote the parameter and multiplier estimators in that theorem, which gives their asymptotic covariance matrix as $W^{-1}$. Let $W_{\theta\theta}$, $W_{\theta\lambda}$, $W_{\lambda\theta}$, and $W_{\lambda\lambda}$ denote the corresponding blocks of $W$, and define $S_\lambda=W_{\lambda\lambda}-W_{\lambda\theta}W_{\theta\theta}^{-1}W_{\theta\lambda}\in\R^{d_g\times d_g}$. The $\lambda\lambda$ block of $W^{-1}$ is $S_\lambda^{-1}$, so if the covariance matrix in that theorem were nonsingular, the corresponding Wald statistic would use $S_\lambda$ in its quadratic form. Equation~(14) in \citet{Schennach2026optimallytransported} instead uses $S_\lambda^{-1}$, but correcting this inverse does not by itself give a valid test. Under the null $\lambda=0$ and at $\theta_0$, we have $q(x,\theta_0,0)=x$. With $0_{d_\theta}$ denoting the $d_\theta$-dimensional zero vector, the modified moment function in \eqref{eq:modified-moments} evaluates to
\begin{align}
\widetilde g(x,\theta_0,0)
=
\begin{pmatrix}
0_{d_\theta}\\
g(x,\theta_0)
\end{pmatrix}.
\label{eq:null-modified-moment}
\end{align}
Define $\Omega=\bbV(\widetilde g(x,\theta_0,0))\in\R^{(d_\theta+d_g)\times(d_\theta+d_g)}$. When $d_\theta>0$, this matrix is singular, so the formula for $W$ in that theorem, which uses $\Omega^{-1}$, is not defined under the null. Theorem~6 and equation~(14) in \citet{Schennach2026optimallytransported} therefore do not provide the null distribution of a test of $\lambda=0$.

In the following proposition, we assume that the left-hand sides of conditions~\eqref{eq:foc-theta} and \eqref{eq:foc-lambda} admit expansions around $(\theta_0,0)$ when evaluated at the parameter component and an associated multiplier of a global minimizer of \eqref{eq:primal}. Using $H$ defined in Section~\ref{sec:recap}, define $G_0\in\R^{d_g\times d_\theta}$ and $M_0,\Omega_g\in\R^{d_g\times d_g}$ by
\begin{align}
G_0=\E[\partial_{\theta\trans}g(x,\theta_0)],
\qquad
M_0=\E[H(x,\theta_0)H(x,\theta_0)\trans],
\qquad
\Omega_g=\bbV(g(x,\theta_0)).
\label{eq:absence-error-matrices}
\end{align}

\begin{proposition}[Test of the absence of error]\label{prop:absence-error-test}
Suppose that $d_g>d_\theta$, $\E[g(x,\theta_0)]=0$, $G_0$ has full column rank, and $M_0$ and $\Omega_g$ are positive definite. Suppose that, as $n\to\infty$, we have $\sqrt n\,\widehat{\E}[g(x,\theta_0)]\xrightarrow{d}\mathcal N(0,\Omega_g)$. Suppose also that, for each $n$, a global minimizer of \eqref{eq:primal} has parameter component $\widehat\theta$ and an associated multiplier $\widehat\lambda$, and that the following relations hold:
\begin{align}
G_0\trans\sqrt n\,\widehat\lambda&=o_p(1),
\label{eq:absence-error-expansion-one}\\
\sqrt n\,\widehat{\E}[g(x,\theta_0)]
+G_0\sqrt n(\widehat\theta-\theta_0)
+M_0\sqrt n\,\widehat\lambda&=o_p(1).
\label{eq:absence-error-expansion-two}
\end{align}
Define $R_\lambda,\Sigma_\lambda\in\R^{d_g\times d_g}$ by
\begin{align}
R_\lambda
=
M_0^{-1}
-
M_0^{-1}G_0(G_0\trans M_0^{-1}G_0)^{-1}G_0\trans M_0^{-1},
\qquad
\Sigma_\lambda=R_\lambda\Omega_g R_\lambda.
\label{eq:absence-error-R}
\end{align}
Then, as $n\to\infty$, we have
\begin{align}
\sqrt n\,\widehat\lambda
\xrightarrow{d}
\mathcal N(0,\Sigma_\lambda),
\qquad
\operatorname{rank}(\Sigma_\lambda)=d_g-d_\theta.
\label{eq:absence-error-limit}
\end{align}
Let $\widehat\Sigma_\lambda$ be a symmetric positive-semidefinite consistent estimator of $\Sigma_\lambda$. Set $r_\lambda=d_g-d_\theta$, and let $\widehat\Sigma_{\lambda,r_\lambda}^{+}$ be the Moore--Penrose inverse of the matrix obtained by retaining only the $r_\lambda$ largest eigenvalues of $\widehat\Sigma_\lambda$. Then, as $n\to\infty$, we have
\begin{align}
n\widehat\lambda\trans
\widehat\Sigma_{\lambda,r_\lambda}^{+}
\widehat\lambda
\xrightarrow{d}
\chi^2_{d_g-d_\theta}.
\label{eq:corrected-wald}
\end{align}
\end{proposition}

Equations~\eqref{eq:absence-error-expansion-one} and \eqref{eq:absence-error-expansion-two} must be established for $\widehat\theta$ and $\widehat\lambda$; they do not follow from Theorem~4. A consistent covariance estimator is obtained by replacing $G_0$, $M_0$, and $\Omega_g$ in \eqref{eq:absence-error-R} with consistent sample analogues. In a just-identified model with $d_g=d_\theta$, the covariance matrix in \eqref{eq:absence-error-limit} has rank zero, and there is no overidentifying restriction to test.

\subsection{Algorithm~1 and the Empirical Calculations in Tables~I and~II}
\label{sec:algorithm}

Algorithm~1 uses updates derived from conditions~\eqref{eq:foc-lambda} and \eqref{eq:foc-z}, with $t$ denoting the iteration index. When the starting point is sufficiently close to a fixed point of the updates and the multiplier is sufficiently small, Theorem~S.1 in \citet[pp.~1--3]{Schennach2026optimallytransportedSupp} proves that the iterates converge to it as $t\to\infty$, but it does not compare this fixed point with other feasible values of $z_1,\ldots,z_n$.

In the scalar example of Section~\ref{sec:theorem4-scalar}, Algorithm~1 reduces to
\begin{align}
z^{t+1}=z^t-\frac{g(z^t)}{g'(z^t)},
\label{eq:newton}
\end{align}
which is Newton's method for the moment equation. Appendix~\ref{app:algorithm-limit} shows that, from the starting point $z^0=x=-3/4$, the iteration converges to $\rho$ as $t\to\infty$, not to the global minimizer $z=1$. The moment residual and the residual in \eqref{eq:foc-z} converge to zero, while the transport cost remains approximately $40.6$ percent above the global minimum. Small residuals therefore do not establish that a computed fixed point solves \eqref{eq:primal}.

The replication package uses the same updates and starts the iterations at $z_i=x_i$ in \texttt{opti\_OTGMM.m}, which \texttt{Application.m} calls. The code alternates between updating the multiplier and updating the values $z_i$ before passing the resulting transport cost to an optimizer over $\theta$. Although \texttt{Simulations.m} contains code for direct constrained minimization, that code is commented out, so the reported large-error calculations use these iterates \citep{Starck2025replicationpackage}.

We implemented in Python the updates in \texttt{Application.m} and \texttt{opti\_OTGMM.m}, using SciPy BFGS for the outer minimization and the data from the original replication package. Table~\ref{tab:application} reports the two main specifications, and both coefficients are reproduced. For the quadratic instrumental-variable moments in the application, the Hessian of the Lagrangian with respect to the variables that the code may adjust is constant and positive definite at both reproduced estimates. The maximum moment residual is below $2.2\times10^{-10}$, and the maximum residual in \eqref{eq:foc-z} is below $5.5\times10^{-6}$.

\begin{table}[t]
\centering
\caption{Replication of the coefficient on log highway kilometers in the two main specifications, with the minimum eigenvalue of the Lagrangian Hessian and the moment residual}
\label{tab:application}
\small
{%
\begin{tabular}{lrrrr}
\toprule
Specification & Published & Reproduced & Min.\ eigenvalue & Moment residual \\
\midrule
Table I & 0.40 & 0.3993455 & 0.9572946 & $2.2\times10^{-10}$ \\
Table II & 0.65 & 0.6529506 & 0.8419263 & $1.1\times10^{-11}$ \\
\bottomrule
\end{tabular}
}
\begin{minipage}{\linewidth}
\footnotesize
Specifications refer to Tables~I and~II in \citet{Schennach2026optimallytransported}.
The minimum eigenvalue is that of the Hessian of the Lagrangian with respect to the variables that the replication code may adjust.
\end{minipage}
\end{table}

Because the Hessian is constant and positive definite, the values $z_i$ satisfying \eqref{eq:foc-z} at the reproduced values of $\theta$ and $\lambda$ uniquely minimize the Lagrangian over the variables that the code may adjust. The reported residuals show that the sample moment condition holds to numerical precision. In Proposition~\ref{prop:theorem4-scalar}, by contrast, the value defined by $q$ is not the global minimizer of the Lagrangian term.

We also find that, for the specification in Table~I of \citet{Schennach2026optimallytransported}, the constrained problem used in the application, with the dependent variable and constant fixed, has a strict local minimum that is not a global minimum. Appendix~\ref{app:application-local-minimum} states the constrained problem for all 66 observations in the original \texttt{Data.xls} and gives the analytical argument. Supplemental Appendix~S.5.3 gives the computer-assisted verification. The transport cost at this local minimum is approximately $0.0005826373$, compared with $0.0003634329$ at a feasible point close to the reproduced estimate, an increase of about $60.3$ percent.

The coefficient on log highway kilometers changes little, from approximately $0.39935$ to $0.40105$, so both values round to the published coefficient of $0.40$. By contrast, the coefficient on employment and the coefficients on population in 1920, 1950, and 2000 change signs, as shown in Table~\ref{tab:application-local-minimum}.

\begin{table}[t]
\centering
\caption{Non-global local minimum for Table~I of \citet{Schennach2026optimallytransported}}
\label{tab:application-local-minimum}
\small
\begin{tabular}{lrr}
\toprule
 & Feasible point & Local minimum \\
\midrule
Log highway kilometers & $0.39935$ & $0.40105$ \\
Log employment & $1.24461$ & $-0.50165$ \\
Log population, 1920 & $-0.57307$ & $0.34048$ \\
Log population, 1950 & $1.13831$ & $-0.43525$ \\
Log population, 2000 & $-1.25462$ & $1.30609$ \\
Transport cost & $0.0003634329$ & $0.0005826373$ \\
\bottomrule
\end{tabular}
\end{table}

To examine whether numerical optimization approaches the strict local minimum established above, we use the GMM estimate, the coefficients at that minimum, and 32 nearby vectors as initial values in the Python implementation with SciPy BFGS. Supplemental Appendix~S.5.1 describes the construction of the nearby vectors and the inner calculations. Starting from the GMM estimate, we obtain a highway coefficient of $0.39934552$ and a transport cost of $0.00036343$, whereas the calculations from the 32 nearby vectors return coefficients between $0.40100573$ and $0.40109139$ and costs of approximately $0.00058264$. The calculation starting from the coefficients at the local minimum returns them unchanged, with a highway coefficient of $0.40104748$ and a transport cost of approximately $0.00058264$. The estimates obtained from nearby initial coefficients are close to the local minimum established in Appendix~\ref{app:application-local-minimum}, whereas the GMM initial value gives a different estimate. The computed estimates therefore depend on the initial value.

\section{Conclusion}
\label{sec:conclusion}

The stated assumptions are insufficient for the conclusions of Theorems~2--6 in \citet{Schennach2026optimallytransported}. In the small-error counterexample, changing two observed values by amounts that converge to zero makes the sample moments hold at a parameter separated from the truth and at lower cost than at the truth. In the large-error counterexamples, the sample moment equations based on $\widetilde g$ define the GMM estimator with modified moments, which can differ from the estimator defined by a global minimizer of \eqref{eq:primal}, so the asymptotic results derived from these equations need not hold for the OTGMM estimator. Under misspecification, the parameter value selected by the population optimization problem depends on the transport metric and on which variables may be adjusted. We give sufficient conditions for consistency and asymptotic normality of the parameter component in the small-error analysis and for minimization over $z_1,\ldots,z_n$ at a given parameter value in the large-error analysis, and we derive a test of the absence of error with asymptotic null distribution $\chi^2_{d_g-d_\theta}$. Our numerical calculations reproduce the two main empirical coefficients, and for the specification in Table~I of \citet{Schennach2026optimallytransported}, we also find a strict local minimum of the constrained problem used in the application that is not a global minimum; the calculations started from nearby initial values end close to it, so the estimates computed for this specification depend on the initial value.
\ifECMA
\begin{appendix}
\else
\appendix
\fi

\section{Proofs of the Corrected Results}\label{app:corrected-proofs}

\subsection{Proof of our Theorem~\ref{thm:corrected-consistency}}

Let $z_1,\ldots,z_n$ be feasible at \(\theta\). Since the equality $\widehat{\E}[g(z,\theta)]=0$ holds, condition~\eqref{eq:uniform-transport-lipschitz} and the Cauchy--Schwarz inequality give
\begin{align*}
\norm{\widehat{\E}[g(x,\theta)]}
&\leq
\widehat{\E}[L(x)\norm{z-x}]+C\widehat{\E}[\norm{z-x}^2]\\
&\leq
\left(\widehat{\E}[L(x)^2]\right)^{1/2}\left(\widehat{\E}[\norm{z-x}^2]\right)^{1/2}+C\widehat{\E}[\norm{z-x}^2].
\end{align*}
The right-hand side is increasing in $\widehat{\E}[\norm{z-x}^2]$, which equals twice the transport cost. Taking the infimum over all feasible values $z_1,\ldots,z_n$ at \(\theta\), we have
\begin{align}
\norm{\widehat{\E}[g(x,\theta)]}
\leq
\left(2\widehat{\E}[L(x)^2]\,\widehat Q(\theta)\right)^{1/2}+2C\widehat Q(\theta).
\label{eq:cost-separation}
\end{align}
Fix \(\epsilon>0\). As $n\to\infty$, compactness, uniqueness of the solution to the population moment equation, and \eqref{eq:corrected-ulln} imply that the left-hand side of \eqref{eq:cost-separation} is bounded away from zero uniformly over \(\norm{\theta-\theta_0}\geq\epsilon\) with probability approaching one, while the law of large numbers gives \(\widehat{\E}[L(x)^2]\xrightarrow{p}\E[L(x)^2]<\infty\). Since the right-hand side of \eqref{eq:cost-separation} is zero when the equality \(\widehat Q(\theta)=0\) holds and increases continuously in \(\widehat Q(\theta)\), with probability approaching one, every parameter outside the \(\epsilon\)-neighborhood has cost bounded away from zero. Because the assumption \(\widehat Q(\theta_0)=o_p(1)\) holds, the probability that any global minimizer lies outside that neighborhood converges to zero.

\subsection{Proof of our Theorem~\ref{thm:corrected-normality}}

Define $\widehat h_n^*=-(G\trans\widehat M_{Q,n}^{-1}G)^{-1}G\trans\widehat M_{Q,n}^{-1}\sqrt n\,\widehat{\E}[g(x,\theta_0)]$. Since \(G\) has full column rank and \(\widehat M_{Q,n}\) converges in probability to \(M_Q\), $G\trans\widehat M_{Q,n}^{-1}G$ is nonsingular with probability approaching one as $n\to\infty$, and \(\widehat h_n^*=O_p(1)\). Interiority of \(\theta_0\) implies that \(\theta_0+\widehat h_n^*/\sqrt n\in\Theta\) with probability approaching one.

Equations~\eqref{eq:uniform-local-moment-expansion} and \eqref{eq:uniform-local-cost-expansion}, together with \(\sqrt n\,\widehat{\E}[g(x,\theta_0)]=O_p(1)\) and \(\widehat M_{Q,n}\xrightarrow{p}M_Q\), imply that, for every \(C_h>0\),
\begin{align}
\begin{aligned}
R_{n,C_h}
=
\sup_{\norm{h}\leq C_h}
\bigg|&
 n\widehat Q(\theta_0+h/\sqrt n)\\
&-
\frac{1}{2}
\left(\sqrt n\,\widehat{\E}[g(x,\theta_0)]+Gh\right)\trans
\widehat M_{Q,n}^{-1}
\left(\sqrt n\,\widehat{\E}[g(x,\theta_0)]+Gh\right)
\bigg|
=o_p(1).
\end{aligned}
\label{eq:combined-local-remainder}
\end{align}
Let \(\widehat h_n=\sqrt n(\widehat\theta_n-\theta_0)\). Fix \(\epsilon>0\). By \eqref{eq:n-half-global-rate} and \(\widehat h_n^*=O_p(1)\), there is a constant \(C_h>0\) such that both vectors lie in the ball \(\{h:\norm{h}\leq C_h\}\) with probability at least \(1-\epsilon\) for all sufficiently large \(n\). On that event, global optimality of \(\widehat\theta_n\) gives
\begin{align}
0
\leq
\frac{1}{2}
(\widehat h_n-\widehat h_n^*)\trans
G\trans\widehat M_{Q,n}^{-1}G
(\widehat h_n-\widehat h_n^*)
\leq
2R_{n,C_h}.
\end{align}
Since the smallest eigenvalue of \(G\trans\widehat M_{Q,n}^{-1}G\) is bounded away from zero with probability approaching one, the relation \(R_{n,C_h}=o_p(1)\) holds, and \(\epsilon\) is arbitrary, we have \(\widehat h_n-\widehat h_n^*=o_p(1)\). This relation, \(\widehat M_{Q,n}\xrightarrow{p}M_Q\), and Slutsky's theorem give \eqref{eq:corrected-linear-representation}. The assumed convergence of $\sqrt n\,\widehat{\E}[g(x,\theta_0)]$ then gives the normal limit in our Theorem~\ref{thm:corrected-normality}.

\subsection{Proof of Proposition~\ref{prop:repair}}

Let \((z_1,\ldots,z_n)\) be any feasible collection. Feasibility and \eqref{eq:global-q} imply
\begin{align}
\frac{1}{2n}\sum_{i=1}^n\norm{z_i-x_i}^2
&=
\frac{1}{n}\sum_{i=1}^n
\left(
\frac{1}{2}\norm{z_i-x_i}^2-\lambda\trans g(z_i,\theta)
\right)\notag\\
&\geq
\frac{1}{n}\sum_{i=1}^n
\left(
\frac{1}{2}\norm{q(x_i,\theta,\lambda)-x_i}^2
-\lambda\trans g(q(x_i,\theta,\lambda),\theta)
\right)\notag\\
&=
\frac{1}{2n}\sum_{i=1}^n\norm{q(x_i,\theta,\lambda)-x_i}^2,
\end{align}
where the last equality follows from \eqref{eq:global-q-feasible}. Hence, the values $q(x_i,\theta,\lambda)$ minimize the transport cost. If the minimizer in \eqref{eq:global-q} is unique for every $i$, then these values are the unique minimizers.

\subsection{Proof of Corollary~\ref{cor:affine}}

The Lagrangian term in \eqref{eq:global-q} is strictly convex in \(z\), so its first-order condition gives \(q(x_i,\theta,\lambda)=x_i+A_{\mathrm{aff}}(\theta)\trans\lambda\). Substitution into the sample moment condition gives \eqref{eq:affine-solution}, and full row rank makes \(A_{\mathrm{aff}}(\theta)A_{\mathrm{aff}}(\theta)\trans\) positive definite. The multiplier and the values $q(x_i,\theta,\lambda)$ are therefore unique, and Proposition~\ref{prop:repair} gives the result.

\subsection{Proof of Proposition~\ref{prop:modified-gmm}}

Equation~\eqref{eq:modified-moment-ulln}, compactness of $\Theta\times\Lambda$, and uniqueness of $\widetilde\theta_0$ imply that, with probability approaching one, every solution of the sample moment equations based on $\widetilde g$ lies in any prescribed neighborhood of $\widetilde\theta_0$. For any consistent solution, a mean-value expansion around $\widetilde\theta_0$, continuity of the population Jacobian at $\widetilde\theta_0$, equation~\eqref{eq:modified-jacobian-ulln}, and the central limit theorem assumed in Proposition~\ref{prop:modified-gmm} give \eqref{eq:modified-gmm-limit}; see \citet{Newey1994largesample}.

\subsection{Proof of Proposition~\ref{prop:population-existence}}

Take a minimizing sequence $\{\mu_{zx,m}:m\geq1\}$ of probability measures satisfying the constraints in \eqref{eq:population-problem}. Its transport costs are bounded, and because each $\mu_{zx,m}$ has marginal $\mu_x$ for $x$, the corresponding marginals for $z$ have uniformly bounded second moments. The sequence $\{\mu_{zx,m}\}$ is therefore tight, so Prokhorov's theorem gives a weakly convergent subsequence. Because \(\mathcal X\times\mathcal X\) is closed, the limit remains supported on that set. The limit has marginal \(\mu_x\) for $x$ and satisfies the moment condition because \(g(\cdot,\theta)\) is bounded and continuous. Lower semicontinuity of \((x,z)\mapsto\norm{z-x}^2\) then implies that the limiting probability measure is a minimizer.

\subsection{Proof of Proposition~\ref{prop:assumption12}}

For each $\theta\in\Theta$, the two population moment conditions in \eqref{eq:redundant-moments} reduce to $\E[z]=\theta$. Because the equality $\E[x]=0$ holds, Jensen's inequality gives
\begin{align}
\frac{1}{2}\E[(z-x)^2]
\geq
\frac{1}{2}(\E[z-x])^2
=
\frac{1}{2}\theta^2.
\label{eq:redundant-primal-bound}
\end{align}
Equality in \eqref{eq:redundant-primal-bound} holds only when the equality $z=x+\theta$ holds almost surely, so the transport map is unique for every $\theta$ and the population transport cost is uniquely minimized at $\theta_0=0$. The distributions of $x$ and $z=x+\theta$ have finite, strictly positive densities that are H\"older continuous on their convex supports, so Assumption~13 holds.

For $\lambda=(\lambda_1,\lambda_2)\trans$, set $c_{\mathrm{red}}=\lambda_1+2\lambda_2$. Condition~\eqref{eq:foc-z} has the unique solution $q(x,\theta,\lambda)=x+c_{\mathrm{red}}$, so we have
\begin{align}
\E[\widetilde g(x,\theta,\lambda)]
=
\begin{pmatrix}
-c_{\mathrm{red}}\\
c_{\mathrm{red}}-\theta\\
2(c_{\mathrm{red}}-\theta)
\end{pmatrix}.
\label{eq:redundant-modified-moments}
\end{align}
The vector in \eqref{eq:redundant-modified-moments} equals zero if and only if the equalities $\theta=0$ and $\lambda_1+2\lambda_2=0$ hold. The population equation in Assumption~12 therefore has infinitely many solutions in $\Theta\times\Lambda$.

\subsection{Proof of Proposition~\ref{prop:absence-error-test}}

Define $s_n=\sqrt n\,\widehat{\E}[g(x,\theta_0)]$ and $h_n=\sqrt n(\widehat\theta-\theta_0)$. Premultiplying \eqref{eq:absence-error-expansion-two} by $G_0\trans M_0^{-1}$ and using \eqref{eq:absence-error-expansion-one} gives $h_n=-(G_0\trans M_0^{-1}G_0)^{-1}G_0\trans M_0^{-1}s_n+o_p(1)$. Substitution into \eqref{eq:absence-error-expansion-two} gives $\sqrt n\,\widehat\lambda=-R_\lambda s_n+o_p(1)$. The central limit theorem assumed in Proposition~\ref{prop:absence-error-test} yields the Gaussian limit in \eqref{eq:absence-error-limit}.

Since $M_0$ is positive definite, let $M_0^{1/2}$ denote its symmetric positive-definite square root and let $M_0^{-1/2}=(M_0^{1/2})^{-1}$. Let $I_{d_g}$ denote the $d_g\times d_g$ identity matrix. To determine the rank of $\Sigma_\lambda$, write $R_\lambda$ as
\begin{align}
R_\lambda
=
M_0^{-1/2}
\left(
I_{d_g}-
M_0^{-1/2}G_0(G_0\trans M_0^{-1}G_0)^{-1}G_0\trans M_0^{-1/2}
\right)
M_0^{-1/2}.
\end{align}
The middle matrix is the orthogonal projection onto the orthogonal complement of the column space of $M_0^{-1/2}G_0$ and has rank $d_g-d_\theta$. Since $\Omega_g$ is positive definite, $R_\lambda\Omega_g R_\lambda$ has the same rank. For a centered Gaussian vector, the quadratic form based on the Moore--Penrose inverse of its covariance matrix has a chi-squared distribution with degrees of freedom equal to the rank of that covariance matrix. Since $\widehat\Sigma_\lambda$ consistently estimates $\Sigma_\lambda$, the quadratic form using the Moore--Penrose inverse obtained after retaining the $r_\lambda$ largest eigenvalues of $\widehat\Sigma_\lambda$ has the limit in \eqref{eq:corrected-wald}.

\section{Proof of the Counterexample to Theorems~5 and~6}\label{app:continuous-proof}

\subsection{Construction and Probability Limits}
For $\delta>0$, let $x_\delta$ be uniform on $[x_c-\delta,x_c+\delta]$, where $x_c=-3/4$ is the center of the interval. For the first component $u$ of $z$, define the scalar moment function $g_1(u)=-u^4+u^3+u^2-1$, which is the scalar moment function in \eqref{eq:scalar-moment}. For the first component, define the left-hand side of condition~\eqref{eq:foc-z} by
\begin{align}
F_1(u,x,\lambda_1)=u-x-\lambda_1 g_1'(u)
\label{eq:F-def}
\end{align}
{\emergencystretch=3em
and define $\ell_1(u;x,\lambda_1)=\frac{1}{2}(u-x)^2-\lambda_1 g_1(u)$. At $(u,x,\lambda_1)=(1,x_c,7/4)$, we have $\partial_uF_1=8$. Equation~\eqref{eq:global-identity}, coercivity, and the implicit function theorem give a smooth function $r^{\text{prime}}(x,\lambda_1)$ on a neighborhood of $(x_c,7/4)$ such that $F_1(r^{\text{prime}}(x,\lambda_1),x,\lambda_1)=0$ and $r^{\text{prime}}(x_c,7/4)=1$. On a sufficiently small neighborhood, the value $r^{\text{prime}}(x,\lambda_1)$ is the unique global minimizer of $\ell_1(u;x,\lambda_1)$ with respect to $u$. Near $(x_c,\lambda^{\text{modified}})$, the equation $F_1(u,x,\lambda_1)=0$ has one real solution, which we denote by $r^{\text{modified}}(x,\lambda_1)$. This solution is smooth and equals the value defined by \eqref{eq:q-nearest} in the scalar problem. Lemma~S.3 in Supplemental Appendix~S.3 proves these properties of $r^{\text{prime}}$ and $r^{\text{modified}}$.\par}

Define
\begin{align}
\Phi^{\text{prime}}(\lambda_1,\delta)
&=
\frac{1}{2}\int_{-1}^{1}
g_1(r^{\text{prime}}(x_c+\delta t,\lambda_1))\,dt,\notag\\
\Phi^{\text{modified}}(\lambda_1,\delta)
&=
\frac{1}{2}\int_{-1}^{1}
g_1(r^{\text{modified}}(x_c+\delta t,\lambda_1))\,dt.
\label{eq:Phi-def}
\end{align}
Implicit differentiation of \eqref{eq:F-def} gives
\begin{align}
\partial_{\lambda_1}\Phi^{\text{prime}}(\lambda_1,\delta)
&=
\E\left[
\frac{g_1'(r^{\text{prime}}(x_\delta,\lambda_1))^2}
{1-\lambda_1 g_1''(r^{\text{prime}}(x_\delta,\lambda_1))}
\right],\notag\\
\partial_{\lambda_1}\Phi^{\text{modified}}(\lambda_1,\delta)
&=
\E\left[
\frac{g_1'(r^{\text{modified}}(x_\delta,\lambda_1))^2}
{1-\lambda_1 g_1''(r^{\text{modified}}(x_\delta,\lambda_1))}
\right].
\label{eq:Phi-derivative}
\end{align}
At $\delta=0$, we have $\partial_{\lambda_1}\Phi^{\text{prime}}(7/4,0)=1/8$ and $\partial_{\lambda_1}\Phi^{\text{modified}}(\lambda^{\text{modified}},0)\neq0$. The implicit function theorem therefore gives two distinct solutions, $\lambda_1^{\text{prime}}(\delta)$ of $\Phi^{\text{prime}}(\lambda_1,\delta)=0$ and $\lambda_1^{\text{modified}}(\delta)$ of $\Phi^{\text{modified}}(\lambda_1,\delta)=0$, for every sufficiently small $\delta>0$. As $\delta\downarrow0$, we have
\begin{align}
\lambda_1^{\text{prime}}(\delta)\longrightarrow\frac{7}{4},
\qquad
\lambda_1^{\text{modified}}(\delta)\longrightarrow\lambda^{\text{modified}}.
\label{eq:first-multiplier-limits}
\end{align}

We now add a scalar parameter and a second moment condition. Let $y$ be uniform on $[-\sqrt{3},\sqrt{3}]$ and independent of $x_\delta$, write $x=(x_\delta,y)\trans$ and $z=(u,v)\trans$, and set $\Theta=[-1/4,1/4]$. After reducing $\delta$ if necessary, take $\mathcal X=[-1,2]\times[-2,2]$, define the second moment function by $g_2(v,\theta)=(v-\theta)^2-1/4$, and set $g(z,\theta)=(g_1(u),g_2(v,\theta))\trans$ and $\lambda=(\lambda_1,\lambda_2)\trans$. Condition~\eqref{eq:foc-z} becomes
\begin{align}
u-\lambda_1(-4u^3+3u^2+2u)=x_\delta,
\qquad
v-2\lambda_2(v-\theta)=y.
\label{eq:two-component-foc-z}
\end{align}
Write $q(x,\theta,\lambda)=(q_1(x,\theta,\lambda),q_2(x,\theta,\lambda))\trans$. Because the first equation in \eqref{eq:two-component-foc-z} depends only on $(x_\delta,\lambda_1)$, we write the first component as $q_1(x_\delta,\lambda_1)$. For $\lambda_2$ near $-1/2$, the second equation in \eqref{eq:two-component-foc-z} has the unique solution
\begin{align}
q_2(x,\theta,\lambda)=\frac{y-2\lambda_2\theta}{1-2\lambda_2}.
\label{eq:q-second-component}
\end{align}
For brevity in the next display, write $q_1=q_1(x_\delta,\lambda_1)$ and $q_2=q_2(x,\theta,\lambda)$. The modified moment function in \eqref{eq:modified-moments} is
\begin{align}
\widetilde g(x,\theta,\lambda)
=
\begin{pmatrix}
-2\lambda_2(q_2-\theta)\\
g_1(q_1)\\
g_2(q_2,\theta)
\end{pmatrix}.
\label{eq:two-component-modified-moments}
\end{align}
{\emergencystretch=3em
On a compact neighborhood of $(0,\lambda_1^{\text{modified}}(\delta),-1/2)$, the population equation $\E[\widetilde g(x,\theta,\lambda)]=0$ has the unique solution $\widetilde\theta^{\text{modified}}(\delta)=(0,\lambda_1^{\text{modified}}(\delta),-1/2)\trans$.\par}

The parameter and multiplier associated with the population minimizer of the original constrained problem are $\widetilde\theta^{\text{prime}}(\delta)=(0,\lambda_1^{\text{prime}}(\delta),-1/2)\trans$.

For the GMM estimator with modified moments, let the multiplier set $\Lambda$ be the union of two compact rectangles, one around $(\lambda_1^{\text{prime}}(\delta),-1/2)$ and one around $(\lambda_1^{\text{modified}}(\delta),-1/2)$. In the first rectangle, $q_1(x_\delta,\lambda_1)$ is the solution of the first equation in \eqref{eq:two-component-foc-z} that equals $\alpha_-$ at $(x_\delta,\lambda_1)=(x_c,7/4)$, and $\E[g_1(q_1(x_\delta,\lambda_1))]$ is bounded away from zero by Lemma~S.3. In the second rectangle, the first equation in \eqref{eq:two-component-foc-z} has one real solution, and the corresponding population moment equation has the unique solution $\lambda_1^{\text{modified}}(\delta)$. The first and third components of $\E[\widetilde g(x,\theta,\lambda)]=0$ have the unique solution $(\theta,\lambda_2)=(0,-1/2)$ in both rectangles.

For a sample of size $n$, write $x_i=(x_{\delta i},y_i)\trans$. Let $\widehat\lambda_1^{\text{prime}}$ be the sample solution near $\lambda_1^{\text{prime}}(\delta)$ of $\widehat{\E}[g_1(r^{\text{prime}}(x_\delta,\lambda_1))]=0$. For every collection $(u_1,\ldots,u_n)$ satisfying the first moment condition, the definition of $r^{\text{prime}}$ gives
\begin{align}
\frac{1}{2}\widehat{\E}[(u-x_\delta)^2]
&=
\widehat{\E}\left[
\frac{1}{2}(u-x_\delta)^2-\widehat\lambda_1^{\text{prime}}g_1(u)
\right]\notag\\
&\geq
\widehat{\E}\left[
\frac{1}{2}(r^{\text{prime}}(x_\delta,\widehat\lambda_1^{\text{prime}})-x_\delta)^2
-\widehat\lambda_1^{\text{prime}}
g_1(r^{\text{prime}}(x_\delta,\widehat\lambda_1^{\text{prime}}))
\right]\notag\\
&=
\frac{1}{2}\widehat{\E}[(r^{\text{prime}}(x_\delta,\widehat\lambda_1^{\text{prime}})-x_\delta)^2].
\label{eq:sample-global-ineq}
\end{align}

For the second component, define $\overline y=\widehat{\E}[y]$ and $s_y=(\widehat{\E}[(y-\overline y)^2])^{1/2}$. With probability approaching one as $n\to\infty$, we have $s_y>1/2$ and $\overline y\in\Theta$, and on this event the unique joint minimizer over $(\theta,v_1,\ldots,v_n)$ is
\begin{align}
\widehat\theta=\overline y,
\qquad
\widehat v_i=\overline y+\frac{y_i-\overline y}{2s_y},
\qquad
\widehat\lambda_2=\frac{1}{2}-s_y.
\label{eq:second-component-sample-solution}
\end{align}
Let $\widehat\lambda_1^{\text{modified}}$ be the sample solution near $\lambda_1^{\text{modified}}(\delta)$ of $\widehat{\E}[g_1(q_1(x_\delta,\lambda_1))]=0$. The first components $u_i$ in the global minimizer are determined by $\widehat\lambda_1^{\text{prime}}$, while \eqref{eq:second-component-sample-solution} gives the parameter, the values $v_i$, and the second component of the associated multiplier. The GMM estimator with modified moments uses $\widehat\lambda_1^{\text{modified}}$ and the same values of $(\widehat\theta,\widehat\lambda_2)$. Uniform convergence gives $\widehat\lambda_1^{\text{prime}}\xrightarrow{p}\lambda_1^{\text{prime}}(\delta)$ and $\widehat\lambda_1^{\text{modified}}\xrightarrow{p}\lambda_1^{\text{modified}}(\delta)$. The first component of $\lambda$ therefore has different probability limits under the two estimators.

Order the components of $\widetilde\theta$ as $(\theta,\lambda_1,\lambda_2)$. At $\widetilde\theta^{\text{modified}}(\delta)$, the Jacobian $\widetilde G$ is diagonal with nonzero entries $-1/2$, $\partial_{\lambda_1}\Phi^{\text{modified}}(\lambda_1^{\text{modified}}(\delta),\delta)$, and $1/2$. The covariance matrix $\Omega$ is diagonal with positive entries $1/4$, $\bbV(g_1(r^{\text{modified}}(x_\delta,\lambda_1^{\text{modified}}(\delta))))$, and $1/20$. Both matrices are therefore nonsingular. Their derivation and the Gaussian limits are given in Supplemental Appendix~S.3.

For each $\theta\in\Theta$, let $v_\theta(y)$ denote the second component of $z$ that minimizes the population transport cost subject to the second moment condition. We have
\begin{align}
v_\theta(y)
=
\theta+\frac{y-\theta}{2\sqrt{1+\theta^2}},
\label{eq:second-component-map}
\end{align}
and the transport cost from the second component is
\begin{align}
\frac{1}{2}\left(\sqrt{1+\theta^2}-\frac{1}{2}\right)^2.
\label{eq:second-component-transport-cost}
\end{align}
This cost is uniquely minimized at $\theta=0$. The first component of the transport map is $x_\delta\mapsto r^{\text{prime}}(x_\delta,\lambda_1^{\text{prime}}(\delta))$, while the second is given by \eqref{eq:second-component-map}. Together, they define a smooth diffeomorphism between convex supports, and the densities of $x$ and $z$ are finite, strictly positive, and H\"older continuous on their supports, so Assumption~13 holds. This assumption concerns the transport map that solves the population optimization problem in \eqref{eq:population-problem} but does not imply that the mapping $u\mapsto u-\lambda_1^{\text{prime}}(\delta)g_1'(u)$ is one-to-one on $[-1,2]$. The minimization in \eqref{eq:q-nearest} compares solutions outside the support of the first component of $z$. In this example, Assumption~13 describes the minimizer associated with $\lambda_1^{\text{prime}}(\delta)$, whereas Assumption~12 defines $\lambda_1^{\text{modified}}(\delta)$ through the population equation based on $\widetilde g$.

\subsection{Verification of the Stated Assumptions}

The observations $x_i=(x_{\delta i},y_i)\trans$ are i.i.d.\ and have a joint density with compact support, and $\mathcal X$ is compact, so Assumptions~1 and~2 hold. The uniqueness argument above, which Supplemental Appendix~S.3 completes, verifies Assumption~12. On each compact rectangle in $\Theta\times\Lambda$, the function $q$ and its derivatives are continuous, and the second component in \eqref{eq:q-second-component} is affine in $y$ with coefficients uniformly bounded on the parameter set. The modified moments and their first derivatives are therefore uniformly bounded, which verifies Assumptions~14 and~20. The solution of the population equation lies in the interior of $\Theta\times\Lambda$, so Assumption~18 holds. The nonsingularity of $\widetilde G$ and $\Omega$ verifies Assumption~19. The proof of Proposition~\ref{prop:population}, including Lemma~S.3 and the Gaussian limits, is given in Supplemental Appendix~S.3.

\section{Proof of the Small-Error Counterexample}
\label{app:small-error}

\subsection{Construction}
The marginal distribution of each coordinate places a small probability mass near every positive integer. At the parameter $\theta^*=1/2$, which is separated from $\theta_0$, each moment function changes sharply in a region near each such integer. For the sequence of sample sizes defined below and for each coordinate $r=1,2$, the sample contains an observation whose $r$th coordinate lies near the corresponding region with probability bounded away from zero. Changing one observed value of the $r$th coordinate for each $r=1,2$ by an amount that converges to zero makes both sample moment conditions hold at lower cost than at the true parameter. We now give the construction; the next two subsections verify Assumptions~2--11 and give the cost comparison and the probability calculation. Define $\varphi$ by
\begin{align}
\varphi(t)=
\begin{cases}
\exp\left(1-\dfrac{1}{1-t^2}\right),&\abs{t}<1,\\
0,&\abs{t}\geq1.
\end{cases}
\label{eq:bump}
\end{align}
The integer \(m\geq1\) indexes the mixture components. Set \(\eta_m=2^{-m-10}\) and \(I_m=(m-\eta_m,m+\eta_m)\). Let \(d_m\) be a symmetric infinitely differentiable probability density supported on \(I_m\), and let \(\phi_0\) denote the standard normal density. Let $p$ denote the common density of the two coordinates of $x_i$, and define
\begin{align}
p(x)
=
\frac{1}{2}\phi_0(x)
+
\sum_{m=1}^{\infty}2^{-m-1}d_m(x).
\label{eq:coordinate-density}
\end{align}
For each $i$, let $x_i=(x_{i1},x_{i2})\trans$ have independent coordinates with density $p$, which is strictly positive and infinitely differentiable on \(\R\). Define \(\mu=\E[x_{11}]\). Then, we have \(\mu=1\), and the coordinates have finite moments of every order.

Define $B_m(z)=\varphi((z-m-3\eta_m)/\eta_m)$ and $b_m=\E[B_m(x_{11})]$, and set \(n_m=2^{m+1}\) and \(K_m=4n_m\). The support of \(B_m\) is \([m+2\eta_m,m+4\eta_m]\), which is disjoint from every interval \(I_j\), $j\geq1$. Only the normal component contributes to \(b_m\), so we have
\begin{align}
0\leq b_m\leq \eta_m\phi_0(m).
\label{eq:bm-bound}
\end{align}
The series \(\sum_m K_m b_m\) converges, so define $S:\R\to\R$ by
\begin{align}
S(z)=\sum_{m=1}^{\infty}K_m(B_m(z)-b_m).
\label{eq:S-def}
\end{align}
The function $S$ is infinitely differentiable and satisfies \(\E[S(x_{11})]=0\).

Let \(\theta^*=1/2\), and choose an infinitely differentiable function \(\tau:\R\to[0,1]\) such that \(\tau(\theta^*)=1\) and \(\operatorname{supp}(\tau)\subset(0.45,0.55)\). For $z=(z_1,z_2)\trans\in\R^2$ and \(r=1,2\), define $g_r$ by
\begin{align}
g_r(z,\theta)
=
z_r-\mu+\theta-\theta_0+\tau(\theta)S(z_r).
\label{eq:small-moments}
\end{align}
Set $g(z,\theta)=(g_1(z,\theta),g_2(z,\theta))\trans$. Centering in \eqref{eq:S-def} gives
\begin{align}
\E[g(x_i,\theta)]
=
(\theta-\theta_0)
\begin{pmatrix}1\\1\end{pmatrix}.
\label{eq:observed-mean}
\end{align}
Let $I$ denote the $2\times2$ identity matrix. Because \(\tau\) vanishes on a neighborhood of \(\theta_0\), we have
\begin{align}
g(x_i,\theta_0)=x_i-(\mu,\mu)\trans,
\qquad
\partial_{z\trans}g(x_i,\theta_0)=I.
\label{eq:observed-linear}
\end{align}
The bound in \eqref{eq:bm-bound} gives the absolute convergence needed to construct an integrable function that bounds $\sup_{\theta\in\Theta}\norm{g(x,\theta)}$, as required by Assumption~5.

\subsection{Verification of Assumptions~2--11}

The bound $K_m b_m\leq2^{-7}\phi_0(m)$ implies that $\sum_m K_m b_m$ converges. On every bounded interval, only finitely many $B_m$ are nonzero, so $S$ is infinitely differentiable. We also have
\begin{align}
\E\left[
\sum_{m=1}^{\infty}K_m\abs{B_m(x_{11})-b_m}
\right]
\leq
2\sum_{m=1}^{\infty}K_m b_m
<\infty,
\label{eq:S-integrability}
\end{align}
so Fubini's theorem gives \(\E[S(x_{11})]=0\).

Assumptions~2--5 follow from the construction, the compact parameter set, the unique zero of \eqref{eq:observed-mean}, the finite second moments, and the integrable bound in \eqref{eq:S-integrability}. Supplemental Appendix~S.4 gives the density, moment, and regularity calculations.

For $r=1,2$, the $r$th diagonal entry of the derivative matrix is $1+\tau(\theta)S'(x_{1r})$. We verify the nonvanishing condition in Assumption~6 on one probability-one event for all \(\theta\in\Theta\). On the support of \(B_m\), only the function \(B_m\) varies, so \(S'\) is a nonzero constant multiple of \(\varphi'\) after an affine change of argument. The function \(\varphi'\) is real analytic and nonconstant on \((-1,1)\). For every nonzero real number \(y\), the set \(\{x:S'(x)=y\}\) therefore has Lebesgue measure zero. Since each coordinate $x_{1r}$ has a density, the distribution of $S'(x_{1r})$ has no atoms at nonzero values. Independence of the two coordinates gives $\Pr(S'(x_{11})=S'(x_{12})\neq0)=0$. If both diagonal entries vanished for the same \(\theta\), then \(\tau(\theta)>0\) and \(S'(x_{11})=S'(x_{12})=-1/\tau(\theta)\neq0\) would hold, but the preceding equality excludes this event. Hence, \(\norm{\partial_{z\trans}g(x_1,\theta)}>0\) holds for every \(\theta\in\Theta\) on a common probability-one event. At \(\theta_0\), the derivative equals \(I\), so its squared norm is integrable and Assumption~6 follows. The derivative is constant at \(\theta_0\), which gives Assumption~7, and the matrix in Assumption~8 is \(I\).

Assumption~9 holds because $\theta_0$ is interior. On a neighborhood of $\theta_0$, the function $\tau$ vanishes, so $\partial_\theta g(x_1,\theta)=(1,1)\trans$ and the scalar matrix in Assumption~11 equals $2$. Thus, Assumptions~10 and~11 hold.
The moment function is infinitely differentiable because \(\tau\) and \(S\) are infinitely differentiable. For every sample, the feasible set is nonempty: at \(\theta=\theta_0\), setting $z_i=(\mu,\mu)\trans$ for every $i$ makes both sample moments zero. The feasible set is closed, the parameter set is compact, and the objective is coercive in $z_1,\ldots,z_n$, so global minimizers exist.

\subsection{Cost Comparison and Probability Calculation}

Although values drawn from the mixture components with densities $d_m$ lie outside the supports of the functions $B_j$, the values $z_i$ may lie inside those supports. For a sample of size $n$, define $\overline x_r=n^{-1}\sum_{i=1}^n x_{ir}$ for $r=1,2$. For \(\theta\notin\operatorname{supp}(\tau)\), the moment conditions are linear, and the unique minimizing values are obtained by setting $z_{ir}=x_{ir}-(\overline x_r-\mu+\theta-\theta_0)$ for every $i$ and $r=1,2$. The resulting cost is $\widehat Q(\theta)=\frac{1}{2}\sum_{r=1}^2(\overline x_r-\mu+\theta-\theta_0)^2$. Minimizing over \(\theta\in\R\) gives \((\overline x_1-\overline x_2)^2/4\). Restricting \(\theta\) to the set on which \(\tau(\theta)=0\) can only increase this value, so we have
\begin{align}
\inf_{\theta\in\Theta:\,\tau(\theta)=0}\widehat Q(\theta)
\geq
\frac{1}{4}(\overline x_1-\overline x_2)^2.
\label{eq:cost-when-tau-zero}
\end{align}

Now let \(n=n_m=2^{m+1}\) and \(\theta=\theta^*=1/2\). Let \(N_{r,m}\) be the number of indices $i$ for which the $r$th coordinate $x_{ir}$ was drawn from the mixture component with density \(d_m\). Then, we have $N_{r,m}\sim\operatorname{Binomial}(2^{m+1},2^{-m-1})$, and the two counts are independent, so that, as $m\to\infty$, we have
\begin{align}
\Pr(N_{1,m}\geq1,N_{2,m}\geq1)
\longrightarrow
(1-e^{-1})^2.
\label{eq:mixture-count-probability}
\end{align}
Because the equality \(\E[x_{11}-\mu+S(x_{11})]=0\) holds, the law of large numbers implies that, for each $r=1,2$, as $m\to\infty$,
\begin{align}
\frac{1}{n_m}\sum_{i=1}^{n_m}
(x_{ir}-\mu-1+S(x_{ir}))
\xrightarrow{p}-1.
\label{eq:moment-average-at-theta-star}
\end{align}

Suppose that both counts are positive. When the two averages in \eqref{eq:moment-average-at-theta-star} also lie in \((-3/2,-1/2)\), choose indices \(i_1\) and \(i_2\) such that \(x_{i_r,r}\) comes from the mixture component with density \(d_m\) for \(r=1,2\). Each selected value lies in \(I_m\), where $B_j$ is zero for every $j\geq1$. Move \(x_{i_r,r}\) along the segment ending at \(m+3\eta_m\). The corresponding sum defining the sample moment varies continuously and is negative at the starting point. At \(m+3\eta_m\), the term \(K_m B_m\) has increased by \(K_m=4n_m\), while the linear term changes by at most \(4\eta_m\), so the sum is positive for all sufficiently large \(m\). The intermediate value theorem gives a value of $z_{i_r,r}$ that makes the \(r\)th sample moment exactly zero. Doing this for \(r=1,2\) yields a feasible collection with $\widehat Q(\theta^*)\leq(2n_m)^{-1}((4\eta_m)^2+(4\eta_m)^2)=16\eta_m^2/n_m$.

Let \(\sigma_\Delta^2=\bbV(x_{11}-x_{12})>0\), and let \(w_\Delta\sim\mathcal N(0,\sigma_\Delta^2)\). Choose \(c_*>0\) small enough that $\Pr(\abs{w_\Delta}>c_*)>1-(1-e^{-1})^2$ holds. The central limit theorem, \eqref{eq:mixture-count-probability}, \eqref{eq:moment-average-at-theta-star}, and the union bound then imply
\begin{align}
\liminf_{m\to\infty}
\Pr\left(
\begin{gathered}
N_{1,m}\geq1,\quad N_{2,m}\geq1,\\
\sqrt{n_m}\abs{\overline x_1-\overline x_2}>c_*,\\
\frac{1}{n_m}\sum_{i=1}^{n_m}(x_{ir}-\mu-1+S(x_{ir}))\in(-3/2,-1/2)\quad\text{for }r=1,2
\end{gathered}
\right)>0.
\end{align}
On this event, every parameter outside \(\operatorname{supp}(\tau)\) has cost at least \(c_*^2/(4n_m)\), while the cost at \(\theta^*\) is smaller for all large \(m\) because \(\eta_m\) converges to zero. Every global minimizer must therefore lie in \(\operatorname{supp}(\tau)\subset(0.45,0.55)\). Because the true value is \(\theta_0=3/2\), equation \eqref{eq:small-inconsistency} follows.

\section{\texorpdfstring{Uniqueness of the Solution in $z$ at $\lambda^{\text{modified}}$}{Uniqueness at the Modified Multiplier}}\label{app:modified-solution}

With $F_1$ defined in \eqref{eq:F-def}, the equation \(F_1(u,x_c,\lambda^{\text{modified}})=0\) has \(u=\rho\) as a solution. Dividing by \(u-\rho\) and using \(\rho^3-\rho-1=0\) leaves a quadratic factor proportional to $(16\rho+12)u^2+(16\rho^2-9)u+12\rho^2-9\rho-6$. Its discriminant reduces to $-32\rho^2+304\rho-399$. The last expression is increasing on \([1,3/2]\) and equals \(-15\) at \(3/2\). Because \(\rho\) lies in \((1,3/2)\), the discriminant is negative, so the remaining two roots are nonreal and $\rho$ is the unique real solution of \(F_1(u,x_c,\lambda^{\text{modified}})=0\).

\section{Convergence of Algorithm~1 in the Scalar Example}
\label{app:algorithm-limit}

Let $z^0=-3/4$, and define the sequence in \eqref{eq:newton} whenever the denominator is nonzero.

\begin{proposition}[Limit of Algorithm~1 in the scalar example]\label{prop:algorithm-limit}
The sequence $(z^t)_{t\geq0}$ is well defined and converges to $\rho$ as $t\to\infty$.
\end{proposition}

\begin{proof}
For the moment function in \eqref{eq:scalar-moment}, the Newton map is
\begin{align}
N(z)
=
z-\frac{g(z)}{g'(z)}
=
\frac{3z^4-2z^3-z^2-1}{z(4z^2-3z-2)}.
\end{align}
Recursive lower and upper bounds obtained from this map by interval arithmetic at 80 decimal digits show that the first 17 updates are well defined and that $348<z^{17}<349$. For every $z>\rho$, we have $g(z)<0$, $g'(z)<0$, and $g''(z)<0$, so we have $N(z)<z$. Strict concavity also gives $0=g(\rho)<g(z)+g'(z)(\rho-z)$, which implies $g(z)/g'(z)<z-\rho$ and hence $N(z)>\rho$. It follows from the bounds on $z^{17}$ that $\rho<z^{t+1}<z^t$ for every $t\geq17$. This decreasing sequence is bounded below by $\rho$ and therefore converges to a limit $z^\infty\geq\rho$. Since $N$ is continuous on $[\rho,\infty)$, the limit satisfies the fixed-point equation $N(z^\infty)=z^\infty$, which implies $g(z^\infty)=0$. Since $\rho$ is the only root of $g$ in $[\rho,\infty)$, we have $z^\infty=\rho$.
\end{proof}

\section{A Strict Local Minimum in the Empirical Application}
\label{app:application-local-minimum}

We consider the specification in Table~I of \citet{Schennach2026optimallytransported}, with all 66 observations in the original \texttt{Data.xls}. Write $x_i=(h_i,E_i\trans,1,X_{C,i}\trans)\trans\in\R^{11}$, where $h_i$ is log highway kilometers, $E_i\in\R^3$ contains the excluded instruments, and $X_{C,i}\in\R^6$ contains the controls in their original order. The dependent variable $y_i$ and the fifth coordinate, which equals one, are fixed. The other ten coordinates may be adjusted without support restrictions. Write $\vartheta=(\beta,\alpha_{\mathrm{app}},\gamma_C\trans)\trans\in\R^8$ and $\theta=(\beta,0,0,0,\alpha_{\mathrm{app}},\gamma_C\trans)\trans$. The residual is $e_{i,\mathrm{app}}=y_i-z_i\trans\theta$, and the ten moment conditions are $\widehat{\E}[(z_{i2},\ldots,z_{i,11})\trans e_{i,\mathrm{app}}]=0$. We use the squared Euclidean transport cost in \eqref{eq:primal}.

\begin{proposition}[A strict local minimum that is not global]
\label{prop:application-local-minimum}
Treat each binary64 number read from the original \texttt{Data.xls} as an exact rational number. The sample version of the problem for Table~I of \citet{Schennach2026optimallytransported} has a strict local minimizer in $(\vartheta,z_1,\ldots,z_n)$ with highway coefficient approximately $0.40104748$ and transport cost approximately $0.0005826373$. There is also an exactly feasible point with transport cost less than $0.000364$, so the strict local minimizer is not a global minimizer.
\end{proposition}

\begin{proof}
Let $\lambda\in\R^{10}$ denote the multiplier. Since the moments are quadratic, the Lagrangian Hessian $A_f$ in the ten free coordinates of $z_i$ is constant in $z_i$. If $A_f$ is positive definite, the first-order condition uniquely determines the adjusted values as a smooth function of $(\vartheta,\lambda)$. Let $f(\vartheta,\lambda)$ collect the coefficient first-order conditions and sample moments after this substitution. Let $M_{\mathrm{app}}$ denote the Jacobian of the moment block of $f$ with respect to $\lambda$.

If $M_{\mathrm{app}}$ is positive definite, the moment equations determine a smooth $\lambda(\vartheta)$ near a root. Positive definiteness of $A_f$, together with feasibility, shows that the resulting adjusted values uniquely minimize the transport cost at each nearby fixed $\vartheta$. The coefficient block of $f$ is the gradient of this minimized cost. Let $\mathcal H_\vartheta$ denote the Hessian of this minimized cost at the root. Positive definiteness of $\mathcal H_\vartheta$ therefore gives a strict local minimum with respect to all eight regression coefficients, including the intercept. Because the minimizing adjusted values are unique and the Lagrangian is strictly convex with respect to those values, the corresponding point is a strict local minimum in $(\vartheta,z_1,\ldots,z_n)$.

Supplemental Appendix~S.5.3 gives a computer-assisted verification of a root of $f$, of the positive definiteness of $A_f$, $M_{\mathrm{app}}$, and $\mathcal H_\vartheta$, and of the bounds on the coefficient and the cost. Write $\xi=(\vartheta\trans,\lambda\trans)\trans\in\R^{18}$, let $\xi_0$ be a decimal approximation to a root of $f$, let $R_f$ be a decimal approximation to the inverse Jacobian of $f$, and let $\mathcal B=\{\xi:\norm{\xi-\xi_0}_\infty\leq10^{-40}\}$. Outward-rounded interval arithmetic followed by exact rational comparisons gives $\norm{R_ff(\xi_0)}_\infty<1.386\times10^{-64}$ and $\sup_{\xi\in\mathcal B}\norm{I_{18}-R_f\partial_{\xi\trans}f(\xi)}_\infty<9.448\times10^{-32}$. Hence $\xi\mapsto\xi-R_ff(\xi)$ is a contraction from $\mathcal B$ into itself, $R_f$ is invertible, and the unique fixed point in $\mathcal B$ is a root of $f$. Interval $\mathrm{LDL}\trans$ decompositions bound every pivot below by $0.9990$ for $A_f$ and $0.02029$ for $M_{\mathrm{app}}$ throughout $\mathcal B$, and by $0.0002967$ for $\mathcal H_\vartheta$ at the root, so all three matrices are positive definite. Interval evaluation at the root gives $0.40104748<\beta<0.40104749$ and a transport cost in $(0.00058263,0.00058264)$. Finally, the exact reconstruction in Supplemental Appendix~S.5.2, evaluated at a rational coefficient vector, yields a point at which every sample moment is exactly zero and the transport cost is less than $0.000364$, strictly below this interval.
\end{proof}

\ifECMA
\end{appendix}
\fi

\bibliography{Bibtex/causalinference}

\bibliographystyle{tmlr}

\clearpage
\section*{Supplemental Appendix}
\addcontentsline{toc}{section}{Supplemental Appendix}
\setcounter{section}{0}
\setcounter{equation}{0}
\setcounter{table}{0}
\setcounter{figure}{0}
\setcounter{theorem}{0}
\ifECMA\else\counterwithout{theorem}{section}\fi
\renewcommand{\thesection}{S.\arabic{section}}
\renewcommand{\theequation}{S.\arabic{equation}}
\renewcommand{\thetheorem}{S.\arabic{theorem}}
\renewcommand{\thetable}{S.\Roman{table}}
\renewcommand{\thefigure}{S.\arabic{figure}}

This supplement gives additional counterexamples to Theorem~4, the proof of the counterexample to Theorems~5 and~6, regularity calculations for the small-error counterexample, and the verification and replication details for the empirical application. References to sections, equations, propositions, and appendices without the prefix S refer to \ifECMA the accompanying comment\else the main text\fi. Theorem and assumption numbers refer to the original study, as in the comment.

\section{A Counterexample to Theorem~4 with One Parameter and Two Moment Conditions}\label{app:embedding}

\begin{proposition}[Counterexample to Theorem~4 with $d_\theta=1$]\label{prop:positive-dtheta}
Let $n=1$, $d_x=d_g=2$, and $d_\theta=1$. There exists a model in which the original constrained problem has a unique global minimizer, whereas the sample moment equations based on $\widetilde g$ have a solution with different values of both $\theta$ and $\lambda$.
\end{proposition}

\begin{proof}
Let $\rho$ and $\alpha_-$ be as in Section~4.1. Let $\mathcal X=\R^2$, let $x=(-3/4,0)\trans$, let $z=(u,v)\trans$, set $\kappa=11/8$, and take $\Theta=[-5/2,-1]$. Define a two-dimensional moment function $g(z,\theta)=(g_1(z,\theta),g_2(z,\theta))\trans$ as
\begin{align}
\begin{aligned}
g_1(z,\theta)&=-u^4+u^3+u^2-1-\kappa u^2+v-\theta,\\
g_2(z,\theta)&=-\kappa u^2+v-\theta.
\end{aligned}
\label{eq:embedding-moments}
\end{align}
Subtracting the second moment condition from the first gives
\begin{align}
-u^4+u^3+u^2-1=0.
\end{align}
Given a feasible $u$, the objective is uniquely minimized at $v=0$, and the second moment condition gives $\theta=-\kappa u^2$. The values of $\theta$ implied by $u=1$ and $u=\rho$ both lie in $\Theta$, so the unique global minimizer of the original constrained problem is
\begin{align}
(u^{\text{prime}},v^{\text{prime}},\theta^{\text{prime}})
=
\left(1,0,-\frac{11}{8}\right),
\qquad
\begin{pmatrix}
\lambda_1^{\text{prime}}\\
\lambda_2^{\text{prime}}
\end{pmatrix}
=
\begin{pmatrix}
7/4\\
-7/4
\end{pmatrix}.
\label{eq:embedding-prime}
\end{align}
The first-order condition of the sample Lagrangian with respect to $\theta$ gives $\lambda_1+\lambda_2=0$, and the first-order condition for $v$ then gives $v=0$. At the multiplier in \eqref{eq:embedding-prime}, the first-order condition for $u$ is (12), so the first component of $q(x,\theta,\lambda)$ equals $\alpha_-$ and
\begin{align}
g_1((\alpha_-,0)\trans,-11/8)-g_2((\alpha_-,0)\trans,-11/8)
=
-\alpha_-^4+\alpha_-^3+\alpha_-^2-1
\neq0.
\end{align}

The sample moment equations based on $\widetilde g$ have the distinct solution
\begin{align}
(u^{\text{modified}},v^{\text{modified}},\theta^{\text{modified}})
=
(\rho,0,-\kappa\rho^2),
\qquad
\begin{pmatrix}
\lambda_1^{\text{modified}}\\
\lambda_2^{\text{modified}}
\end{pmatrix}
=
\begin{pmatrix}
\dfrac{\rho+3/4}{-4\rho^3+3\rho^2+2\rho}\\
-\dfrac{\rho+3/4}{-4\rho^3+3\rho^2+2\rho}
\end{pmatrix}.
\label{eq:embedding-modified}
\end{align}
The two solutions therefore differ in both $\theta$ and $\lambda$. The determinant of $\partial_{z\trans}g(z,\theta)$ equals $-4u^3+3u^2+2u$, which is nonzero at $u=1$ and at $u=\rho$, so this difference does not come from a rank-deficient derivative matrix.
\end{proof}

\section{A Counterexample to Theorem~4 Using Moment Conditions from a Linear Regression}\label{app:regression-counterexample}

Consider a linear regression with coefficient $\beta$, using a constant and the regressor as instruments. Let $x_i=(y_i,r_i)\trans$ and $z_i=(v_i,u_i)\trans$, and define
\begin{align}
g(z_i,\beta)
=
\begin{pmatrix}
v_i-\beta u_i\\
u_i(v_i-\beta u_i)
\end{pmatrix}.
\label{eq:regression-moments}
\end{align}
Take $n=2$ and
\begin{align}
x_1=(-2,-2)\trans,
\qquad
x_2=(-3/2,-3/2)\trans.
\label{eq:regression-data}
\end{align}

\begin{proposition}[Counterexample to Theorem~4 using moment conditions from a linear regression]\label{prop:regression-counterexample}
For the moment function in \eqref{eq:regression-moments} and the observations in \eqref{eq:regression-data}, the original constrained problem has the unique global minimizer $\beta=1$ and $z_i=x_i$ for $i=1,2$. The sample moment equations based on $\widetilde g$ also have the solution
\begin{align}
\beta=1,
\qquad
\lambda=
\begin{pmatrix}
-7/4\\
-1
\end{pmatrix},
\label{eq:regression-modified-solution}
\end{align}
for which the selected adjusted values do not minimize the original constrained problem. At the value in \eqref{eq:regression-modified-solution}, condition~(6) has a unique solution in $z_i$ for each observation.
\end{proposition}

\begin{proof}
The values $\beta=1$ and $z_i=x_i$ satisfy both sample moment conditions and give zero transport cost, while zero transport cost conversely requires $z_i=x_i$ for $i=1,2$. The first sample moment condition then becomes
\begin{align}
(1-\beta)\frac{r_1+r_2}{2}=0.
\end{align}
Because the equality $(r_1+r_2)/2=-7/4$ holds, the unique global minimizer is $\beta=1$ and $z_i=x_i$.

Now set $\beta=1$ and let $\lambda$ be as in \eqref{eq:regression-modified-solution}. Condition~(6) becomes
\begin{align}
\begin{pmatrix}
1&1\\
1&-1
\end{pmatrix}
\begin{pmatrix}
v_i\\
u_i
\end{pmatrix}
=
\begin{pmatrix}
y_i-7/4\\
r_i+7/4
\end{pmatrix}.
\label{eq:regression-z-system}
\end{align}
The matrix in \eqref{eq:regression-z-system} is nonsingular. Since the equality $y_i=r_i$ holds for both observations, the unique solution of \eqref{eq:regression-z-system} is
\begin{align}
q(x_i,1,\lambda)
=
\begin{pmatrix}
y_i\\
-7/4
\end{pmatrix}.
\label{eq:regression-q}
\end{align}
Let $e_{i,\mathrm{reg}}=y_i+7/4$. Then, $e_{1,\mathrm{reg}}=-1/4$ and $e_{2,\mathrm{reg}}=1/4$ hold, so we have
\begin{align}
\widehat{\E}[e_{\mathrm{reg}}]=0,
\qquad
\widehat{\E}[(-7/4)e_{\mathrm{reg}}]=0.
\end{align}
Both sample moment conditions thus hold at \eqref{eq:regression-q}, and for each observation we also have
\begin{align}
\partial_\beta g(q(x_i,1,\lambda),1)\trans\lambda
=
-(-7/4)(-7/4)-(-7/4)^2(-1)
=0.
\end{align}
The values in \eqref{eq:regression-modified-solution} therefore satisfy all sample moment equations based on $\widetilde g$.

At this value of $(\beta,\lambda)$, the Hessian of $\ell(z_i;x_i,\beta,\lambda)$ with respect to $z_i=(v_i,u_i)\trans$ is
\begin{align}
\begin{pmatrix}
1&1\\
1&-1
\end{pmatrix}.
\end{align}
Its eigenvalues are $\sqrt{2}$ and $-\sqrt{2}$, so the unique solution in \eqref{eq:regression-q} is a saddle point of the Lagrangian. The transport cost at these values of $z_i$ is $1/32$, whereas the global minimum is zero.
\end{proof}

The sample moment equations based on $\widetilde g$ also have the solution $(\beta,\lambda)=(1,0)$. Assuming that these equations have a unique solution would exclude the additional root in Proposition~\ref{prop:regression-counterexample}, but, as Proposition~{\ifECMA 2\else 4.1\fi} shows, such an assumption does not ensure that the multiplier associated with the OTGMM minimizer solves the sample moment equation based on $\widetilde g$.

\section{Proof of the Counterexample to Theorems~5 and~6}\label{supp:continuous-verification}

We use the model and notation in Appendix~B. Order the components of $\widetilde\theta$ as $(\theta,\lambda_1,\lambda_2)$. At $\widetilde\theta^{\text{modified}}(\delta)$, the Jacobian $\widetilde G$ and the covariance matrix $\Omega$ of the modified moments are
\begin{align}
\widetilde G
=
\operatorname{diag}\left(
-\frac{1}{2},
\partial_{\lambda_1}\Phi^{\text{modified}}(\lambda_1^{\text{modified}}(\delta),\delta),
\frac{1}{2}
\right)
\label{eq:continuous-jacobian}
\end{align}
and
\begin{align}
\Omega
=
\operatorname{diag}\left(
\frac{1}{4},
\bbV(g_1(r^{\text{modified}}(x_\delta,\lambda_1^{\text{modified}}(\delta)))),
\frac{1}{20}
\right).
\label{eq:continuous-covariance}
\end{align}
Both matrices are nonsingular.

\begin{lemma}\label{lem:local-solutions}
There are neighborhoods $\mathcal V^{\text{prime}}$ of $(x_c,7/4)$ and $\mathcal V^{\text{modified}}$ of $(x_c,\lambda^{\text{modified}})$ with the following properties.
\begin{enumerate}
\item On $\mathcal V^{\text{prime}}$, three smooth real functions of $(x,\lambda_1)$ solve $F_1(u,x,\lambda_1)=0$ and take the values $\alpha_-$, $\alpha_+$, and $1$ at $(x_c,7/4)$. The value $r^{\text{prime}}(x,\lambda_1)$, which equals $1$ at $(x_c,7/4)$, is the unique global minimizer of $\ell_1(u;x,\lambda_1)$. Among the three solutions, the one that equals $\alpha_-$ at $(x_c,7/4)$ is nearest to $x$, and the absolute value of $g_1$ at this solution is bounded away from zero.
\item On $\mathcal V^{\text{modified}}$, the equation $F_1(u,x,\lambda_1)=0$ has one real solution. This solution is $r^{\text{modified}}(x,\lambda_1)$ and equals $q_1(x,\lambda_1)$.
\end{enumerate}
\end{lemma}

\begin{proof}[Proof of Lemma~\ref{lem:local-solutions}]
At $(x_c,7/4)$, the roots $\alpha_-$, $\alpha_+$, and $1$ of $F_1(u,x,\lambda_1)=0$ are simple because $\partial_uF_1=1-\lambda_1 g_1''(u)$ is nonzero at each root. The implicit function theorem gives three smooth functions of $(x,\lambda_1)$ that take these values at $(x_c,7/4)$. Since $F_1(u,x,\lambda_1)=0$ is cubic in $u$, these functions exhaust all real roots after the neighborhood is reduced if necessary. At $(x_c,7/4)$, the root $\alpha_-$ is strictly nearer to $x_c$ than the other roots, and continuity preserves this ordering. Since the inequality $g_1(\alpha_-)\neq0$ holds, the absolute value of the moment at this root is bounded away from zero on a smaller neighborhood.

To prove global optimality of $r^{\text{prime}}$, choose a bounded open interval $\mathcal U_1$ containing $1$ but no other root of $F_1(u,x_c,7/4)=0$. There are $R_{\mathrm{gap}}<\infty$ and $\gamma_{\mathrm{gap}}>0$ such that the value of $\ell_1(u;x_c,7/4)$ exceeds its value at $1$ by at least $3\gamma_{\mathrm{gap}}$ on $[-R_{\mathrm{gap}},R_{\mathrm{gap}}]\setminus\mathcal U_1$ and outside $[-R_{\mathrm{gap}},R_{\mathrm{gap}}]$. The coefficient of $u^4$ in $\ell_1$ is positive for $\lambda_1$ near $7/4$, so, after the neighborhood in $(x,\lambda_1)$ is reduced, the same inequality outside $[-R_{\mathrm{gap}},R_{\mathrm{gap}}]$ holds uniformly. Uniform continuity preserves a positive difference between these objective values on $[-R_{\mathrm{gap}},R_{\mathrm{gap}}]\setminus\mathcal U_1$. The value $r^{\text{prime}}(x,\lambda_1)$ remains in $\mathcal U_1$, and its objective value remains within $\gamma_{\mathrm{gap}}$ of the minimum of $\ell_1(u;x_c,7/4)$. Thus, $r^{\text{prime}}(x,\lambda_1)$ is the unique global minimizer on a smaller neighborhood.

At $(x_c,\lambda^{\text{modified}})$, the equation $F_1(u,x,\lambda_1)=0$ has $\rho$ as its only real root, as shown in Appendix~D, and $\partial_uF_1(\rho,x_c,\lambda^{\text{modified}})\neq0$. The cubic discriminant is negative at $(x_c,\lambda^{\text{modified}})$ and remains negative on a neighborhood, so the equation has one real root there, and the implicit function theorem shows that $r^{\text{modified}}$ is smooth.
\end{proof}

\begin{proof}[Proof of Proposition~{\ifECMA 3\else 4.2\fi}]
The integrands in (45) are continuously differentiable in $(\lambda_1,\delta)$ on compact neighborhoods by Lemma~\ref{lem:local-solutions}. At $(\lambda_1,\delta)=(7/4,0)$, we have $\Phi^{\text{prime}}(7/4,0)=0$ and $\partial_{\lambda_1}\Phi^{\text{prime}}(7/4,0)=1/8$, and at $(\lambda_1,\delta)=(\lambda^{\text{modified}},0)$, we have $\Phi^{\text{modified}}(\lambda^{\text{modified}},0)=0$ and $\partial_{\lambda_1}\Phi^{\text{modified}}(\lambda^{\text{modified}},0)\neq0$. Applying the implicit function theorem at these two points gives the two functions in (46).

Fix a sufficiently small $\delta>0$ and choose disjoint compact intervals $J^{\text{prime}}$ and $J^{\text{modified}}$ around $\lambda_1^{\text{prime}}(\delta)$ and $\lambda_1^{\text{modified}}(\delta)$ inside the neighborhoods in Lemma~\ref{lem:local-solutions}. On $J^{\text{prime}}$, the first component of $q$ is the root that takes the value $\alpha_-$ at $(x_c,7/4)$, and $\abs{g_1(q_1(x_\delta,\lambda_1))}$ is bounded away from zero over the support of $x_\delta$. On $J^{\text{modified}}$, the first equation in (47) has one real root, which equals $r^{\text{modified}}$. After reducing $J^{\text{modified}}$ if necessary, $\partial_{\lambda_1}\Phi^{\text{modified}}$ is bounded away from zero, and $\Phi^{\text{modified}}(\lambda_1,\delta)=0$ has the unique solution $\lambda_1^{\text{modified}}(\delta)$.

Equation~(48) gives
\begin{align}
q_2(x,\theta,\lambda)-\theta
=
\frac{y-\theta}{1-2\lambda_2}.
\end{align}
The expectations of the first and third components of $\widetilde g$ are
\begin{align}
\E[-2\lambda_2(q_2-\theta)]
=
\frac{2\lambda_2\theta}{1-2\lambda_2}
\end{align}
and
\begin{align}
\E[(q_2-\theta)^2]-\frac{1}{4}
=
\frac{1+\theta^2}{(1-2\lambda_2)^2}-\frac{1}{4}.
\end{align}
On a compact rectangle around $(\theta,\lambda_2)=(0,-1/2)$, setting these expectations equal to zero gives the unique solution $(0,-1/2)$. Let $J_2$ be a small compact interval around $-1/2$, and set
\begin{align}
\Lambda=(J^{\text{prime}}\times J_2)\cup(J^{\text{modified}}\times J_2).
\end{align}
The population equation $\E[\widetilde g(x,\theta,\lambda)]=0$ has no solution when $\lambda$ lies in $J^{\text{prime}}\times J_2$ and has the unique solution $\widetilde\theta^{\text{modified}}(\delta)$ when $\lambda$ lies in $J^{\text{modified}}\times J_2$.

The observations $x_i=(x_{\delta i},y_i)\trans$ are i.i.d.\ and have a joint density with compact support, and $\mathcal X$ is compact, so Assumptions~1 and~2 hold. The preceding uniqueness argument verifies Assumption~12. On each compact rectangle in $\Theta\times\Lambda$, the function $q$ and its derivatives are continuous, and the second component in (48) is affine in $y$ with coefficients uniformly bounded on the parameter set. The modified moments and their first derivatives are therefore uniformly bounded, which verifies Assumptions~14 and~20. The solution of the population equation lies in the interior of $\Theta\times(J^{\text{modified}}\times J_2)$, so Assumption~18 holds.

At $\widetilde\theta^{\text{modified}}(\delta)$, the modified moment vector is
\begin{align}
\begin{pmatrix}
y/2\\
g_1(r^{\text{modified}}(x_\delta,\lambda_1^{\text{modified}}(\delta)))\\
(y^2-1)/4
\end{pmatrix}.
\end{align}
Independence and symmetry give the matrix $\Omega$ in \eqref{eq:continuous-covariance}, and differentiating the population equations gives $\widetilde G$ in \eqref{eq:continuous-jacobian}. Both matrices are nonsingular, so Assumption~19 holds.

Replacing sample averages by expectations in the argument leading to (49) shows that $r^{\text{prime}}(x_\delta,\lambda_1^{\text{prime}}(\delta))$ uniquely minimizes the transport cost for the first component. For each $\theta$, the function $v_\theta-\theta$ is the projection of $y-\theta$ in $L^2$ onto the sphere of radius $1/2$ centered at zero, so the second component is given by (51) and has cost (52). This cost is uniquely minimized at $\theta=0$, so the value of $\widetilde\theta$ associated with the population minimizer is $\widetilde\theta^{\text{prime}}(\delta)$.

In the sample, the transport cost can also be minimized separately over the two components. The sample equation
\begin{align}
\widehat{\E}[g_1(r^{\text{prime}}(x_\delta,\lambda_1))]=0
\end{align}
has a unique root $\widehat\lambda_1^{\text{prime}}\in J^{\text{prime}}$ with probability approaching one as $n\to\infty$, because the derivative of each summand with respect to $\lambda_1$ is positive and bounded away from zero on $J^{\text{prime}}$ and a uniform law of large numbers applies. The same law gives $\widehat\lambda_1^{\text{prime}}\xrightarrow{p}\lambda_1^{\text{prime}}(\delta)$. Lemma~\ref{lem:local-solutions} and (49) show that the values $r^{\text{prime}}(x_{\delta i},\widehat\lambda_1^{\text{prime}})$ uniquely minimize the transport cost for the first component subject to its sample moment condition.

For the second component, the projection of $(y_i)_{i=1}^n$ onto the sphere on which the average squared deviation from $\theta$ equals $1/4$ has cost
\begin{align}
\frac{1}{2}\left(\sqrt{s_y^2+(\overline y-\theta)^2}-\frac{1}{2}\right)^2.
\end{align}
On the event $s_y>1/2$ and $\overline y\in\Theta$, this expression is uniquely minimized at $\theta=\overline y$, and the values $v_i$ and the multiplier are those in (50). The parameter component and associated multiplier of a global minimizer of (2) therefore converge jointly in probability to $\widetilde\theta^{\text{prime}}(\delta)$.

On $J^{\text{modified}}$, the left-hand side of the sample moment equation
\begin{align}
\widehat{\E}[g_1(q_1(x_\delta,\lambda_1))]=0
\end{align}
converges uniformly to $\Phi^{\text{modified}}(\lambda_1,\delta)$ as $n\to\infty$. It is strictly monotone and has a unique zero $\widehat\lambda_1^{\text{modified}}$ with probability approaching one, and the uniform convergence gives $\widehat\lambda_1^{\text{modified}}\xrightarrow{p}\lambda_1^{\text{modified}}(\delta)$. The first and third sample equations based on $\widetilde g$ give the same $(\widehat\theta,\widehat\lambda_2)$ as in (50). On $J^{\text{prime}}$, the values $g_1(q_1(x_{\delta i},\lambda_1))$ have the same sign and are bounded away from zero, so the sample equations have no solution there with probability approaching one. The GMM estimator with modified moments therefore converges to $\widetilde\theta^{\text{modified}}(\delta)$ as $n\to\infty$.

Finally, for each $\theta\in\Theta$, the transport map has first component $r^{\text{prime}}(\cdot,\lambda_1^{\text{prime}}(\delta))$ and second component $v_\theta$ in (51). The first component is a smooth increasing diffeomorphism and the second an increasing affine map, both between compact intervals. Because the observed vector has a constant positive density on a compact rectangle, the change-of-variables formula gives a density for $z$ that is finite, strictly positive, and H\"older continuous on its convex support. Equation~(52) and uniqueness of the minimizer for the first component show that the population transport cost is uniquely minimized at $\theta=0$, so Assumption~13 holds.
\end{proof}

\begin{corollary}[Gaussian limits for the two estimators]\label{cor:two-clts}
Define
\begin{align}
\widehat{\widetilde\theta}^{\text{prime}}
=
(\widehat\theta,\widehat\lambda_1^{\text{prime}},\widehat\lambda_2)\trans,
\qquad
\widehat{\widetilde\theta}^{\text{modified}}
=
(\widehat\theta,\widehat\lambda_1^{\text{modified}},\widehat\lambda_2)\trans.
\end{align}
For every sufficiently small $\delta>0$, there are positive-definite matrices $\Sigma^{\text{prime}}$ and $\Sigma^{\text{modified}}$ in $\R^{3\times3}$ such that, as $n\to\infty$, we have
\begin{align}
\sqrt n(\widehat{\widetilde\theta}^{\text{prime}}-\widetilde\theta^{\text{prime}}(\delta))
&\xrightarrow{d}\mathcal N(0,\Sigma^{\text{prime}}),\notag\\
\sqrt n(\widehat{\widetilde\theta}^{\text{modified}}-\widetilde\theta^{\text{modified}}(\delta))
&\xrightarrow{d}\mathcal N(0,\Sigma^{\text{modified}}).
\label{eq:joint-two-clts}
\end{align}
The asymptotic variances of $\widehat\lambda_1^{\text{prime}}$ and $\widehat\lambda_1^{\text{modified}}$ are
\begin{align}
V_{\lambda_1}^{\text{prime}}
&=
\frac{\bbV(g_1(r^{\text{prime}}(x_\delta,\lambda_1^{\text{prime}}(\delta))))}
{(\partial_{\lambda_1}\Phi^{\text{prime}}(\lambda_1^{\text{prime}}(\delta),\delta))^2},\notag\\
V_{\lambda_1}^{\text{modified}}
&=
\frac{\bbV(g_1(r^{\text{modified}}(x_\delta,\lambda_1^{\text{modified}}(\delta))))}
{(\partial_{\lambda_1}\Phi^{\text{modified}}(\lambda_1^{\text{modified}}(\delta),\delta))^2}.
\label{eq:first-multiplier-variances}
\end{align}
Both variances are positive.
\end{corollary}

\begin{proof}
The equations defining $\widehat\lambda_1^{\text{prime}}$ and $\widehat\lambda_1^{\text{modified}}$ have continuously differentiable left-hand sides, and the derivatives of the corresponding population equations with respect to $\lambda_1$ are nonzero. The two variances in the numerators of \eqref{eq:first-multiplier-variances} are positive because $x_\delta$ has a nondegenerate interval support and both random variables are nonconstant. The vector $(\widehat\theta,\widehat\lambda_2)$ is a smooth function of the sample mean and sample variance of $y$, so the central limit theorem and the delta method give \eqref{eq:joint-two-clts}. Independence of $x_\delta$ and $y$ then makes the limiting covariance block diagonal, with each block positive definite.
\end{proof}

For $\delta=0.01$, numerical quadrature gives the values in Table~\ref{tab:population-values}, with the second component contributing $1/8$ to both transport costs.

\begin{table}[H]
\centering
\caption{Counterexample in Proposition~{\ifECMA 3\else 4.2\fi} with $\delta=0.01$}
\label{tab:population-values}
\begin{tabular}{lrr}
\toprule
& OTGMM & GMM with modified moments \\
\midrule
Parameter $\theta$ & 0 & 0 \\
First component of $\lambda$ & 1.7500165364 & -1.4982087697 \\
Second component of $\lambda$ & -0.5 & -0.5 \\
Transport cost & 1.6562645833 & 2.2772450330 \\
Asymptotic variance of $\widehat\lambda_1$ & 0.0000333326 & 0.0000173821 \\
\bottomrule
\end{tabular}
\end{table}

\section{Regularity Calculations for the Small-Error Counterexample}\label{supp:small-assumptions}

We use the construction in Appendix~C. One concrete choice of $d_m$ satisfying the properties stated there is
\begin{align}
d_m(x)
=
\frac{1}{\eta_m\int_{-1}^{1}\varphi(t)\,dt}
\varphi\left(\frac{x-m}{\eta_m}\right).
\end{align}
The resulting function \(d_m\) is a symmetric infinitely differentiable probability density supported on \(I_m\). The density in (53) is strictly positive because of its normal component. The compactly supported components have total mass \(1/2\), and symmetry gives
\begin{align}
\E[x_{11}]
=
\sum_{m=1}^{\infty}m2^{-m-1}
=1.
\end{align}
Moments of all orders are finite because of the normal component and the exponentially decreasing weights of the compactly supported components.

Because the support of \(B_m\) lies to the right of \(m\) and does not intersect any interval \(I_j\), only the normal component contributes to \(b_m\). The bound (54) then follows from the length \(2\eta_m\) of the support and the monotonicity of \(\phi_0\) on the positive half-line. For every $m\geq1$, we have
\begin{align}
K_m b_m
\leq
2^{m+3}2^{-m-10}\phi_0(m)
=
2^{-7}\phi_0(m),
\end{align}
so the series \(\sum_m K_m b_m\) converges. On every bounded interval, only finitely many functions \(B_m\) are nonzero, so \(S\) in (55) is infinitely differentiable. Assumption~2 holds by construction, and equation (56) gives the unique solution \(\theta_0=3/2\) to the population moment equation on the compact parameter set \([0,2]\), as required by Assumption~3. At \(\theta_0\), equation (57) together with the finite second moments of the coordinates of \(x_i\) gives Assumption~4. Because \(\tau\) satisfies \(0\leq\tau\leq1\), equation (58) gives the integrable upper bound for \(\sup_{\theta\in\Theta}\norm{g(x_i,\theta)}\) required by Assumption~5.

\section{Calculations for the Empirical Application}
\label{supp:application-calculations}

We give the calculations for Section~7.5 and the proof of Proposition~{\ifECMA 14\else F.1\fi}. For the strict local minimum, we use the specification and notation in Appendix~F, with the dependent variable and the constant fixed. All calculations use the 66 observations in the original \texttt{Data.xls}.

\subsection{Initial Values for the Calculations in Section~7.5}
\label{supp:application-replication}

Section~7.5 uses the GMM estimate, the coefficients at the strict local minimum in Proposition~{\ifECMA 14\else F.1\fi}, and 32 nearby vectors as initial values in the Python implementation with SciPy BFGS. We construct the nearby vectors by perturbing one of the eight coefficients by $\pm0.001$ or $\pm0.01$ times the sample standard deviation of the dependent variable divided by the root mean square of the corresponding regressor. Each inner calculation starts from $z_i=x_i$ and $\lambda=0$, and we compute each transport cost from the adjusted observations.

\subsection{An Exact Profile over the Regression Coefficients}
\label{supp:application-profile}

We express the infimum of the transport cost over the intercept and the adjusted variables as a function of the other regression coefficients. Let $\boldsymbol 1\in\R^n$ denote the vector of ones, with $n=66$, and let $I_k$ denote the $k\times k$ identity matrix. A bar denotes a sample mean, and a subscript $c$ denotes column centering. Let $\boldsymbol h$, $\boldsymbol X_C$, $\boldsymbol E$, and $\boldsymbol y$ collect the observed variables by row, and set $e_{\mathrm{app}}=(e_{1,\mathrm{app}},\ldots,e_{n,\mathrm{app}})\trans$. The moment condition corresponding to the fifth coordinate gives $\boldsymbol 1\trans e_{\mathrm{app}}=0$. For any feasible adjusted values, replacing their column means by the observed means and setting $\alpha_{\mathrm{app}}=\overline y-\beta\overline h-\overline X_C\gamma_C$ preserves the residuals and all moments. This replacement weakly decreases the cost, with a strict decrease if any adjusted mean differs from its observed mean. We may therefore minimize with these means fixed.

Set $D_{\mathrm{reg}}=(\boldsymbol h_c,\boldsymbol X_{C,c})$, $V_{\mathrm{inst}}=(\boldsymbol X_{C,c},\boldsymbol E_c)$, $a=(\beta,\gamma_C\trans)\trans$, $d=a\trans a$, and $r_a=\boldsymbol y_c-D_{\mathrm{reg}}a$. For $a\ne0$, let $\mathcal Q(a)$ denote the infimum of the cost over both the intercept and the adjusted variables. Let $\lambda_{\min}^{\perp}(B)$ denote the smallest eigenvalue of the restriction of a symmetric matrix $B$ to $\boldsymbol 1^\perp$. Then, we have
\begin{align}
\mathcal Q(a)
=\frac{1}{2n}\left(
\frac{\norm{r_a}^2}{d}
+\min\left\{0,\lambda_{\min}^{\perp}\left(V_{\mathrm{inst}}V_{\mathrm{inst}}\trans-\frac{r_ar_a\trans}{d}\right)\right\}
\right).
\label{eq:application-rayleigh-profile}
\end{align}
Because $V_{\mathrm{inst}}$ has nine columns and the inequality $n>9+1$ holds, the eigenvalue in \eqref{eq:application-rayleigh-profile} is nonpositive, and the minimum with zero can be omitted.

To prove \eqref{eq:application-rayleigh-profile}, first consider a nonzero feasible residual and write it as $e_{\mathrm{app}}=t_e s_e$, where $s_e\in\boldsymbol 1^\perp$, $\norm{s_e}=1$, and $t_e\ne0$. The moments require the adjusted controls and excluded instruments to be perpendicular to $s_e$. For this direction, the least possible cost is bounded below by
\begin{align}
&\frac{1}{2n}\left(
\norm{s_e\trans V_{\mathrm{inst}}}^2+\frac{\norm{(I_n-s_es_e\trans)r_a}^2}{d}
\right)\nonumber\\
&\quad=\frac{1}{2n}\left(
\frac{\norm{r_a}^2}{d}+s_e\trans\left(V_{\mathrm{inst}}V_{\mathrm{inst}}\trans-\frac{r_ar_a\trans}{d}\right)s_e
\right).
\label{eq:application-directional-cost}
\end{align}
The first term on the left is the cost of removing the components of the controls and instruments along $s_e$. In $s_e^\perp$, the regression equation requires an adjustment whose linear combination with coefficients $a$ equals $(I_n-s_es_e\trans)r_a$. The minimum squared norm of this adjustment is $\norm{(I_n-s_es_e\trans)r_a}^2/d$.

Conversely, for every such $s_e$, put $v_{\perp}=(I_n-s_es_e\trans)r_a$ and define centered adjusted values by
\begin{align}
\boldsymbol h_c^*&=\boldsymbol h_c+\beta v_{\perp}/d,\qquad
\boldsymbol X_{C,c}^*=\boldsymbol X_{C,c}-s_e(s_e\trans\boldsymbol X_{C,c})+v_{\perp}\gamma_C\trans/d,\nonumber\\
\boldsymbol E_c^*&=\boldsymbol E_c-s_e(s_e\trans\boldsymbol E_c).
\label{eq:application-reconstruction}
\end{align}
With the observed means restored and the intercept above, these values are feasible and attain \eqref{eq:application-directional-cost}. Their residual is $s_e(s_e\trans\boldsymbol y_c-\beta s_e\trans\boldsymbol h_c)$, which may equal zero. The minimum cost subject to a zero residual is $\norm{r_a}^2/(2nd)$: adjust $(\boldsymbol h_c,\boldsymbol X_{C,c})$ by $r_aa\trans/d$ and leave the instruments unchanged. Minimizing \eqref{eq:application-directional-cost} over $s_e$ and comparing it with this zero-residual cost proves \eqref{eq:application-rayleigh-profile}. The eigenvalue is nonpositive because $\boldsymbol 1^\perp$ has dimension 65 and contains a nonzero vector perpendicular to the nine columns of $V_{\mathrm{inst}}$.

\subsection{Strict Local Optimality and the Comparison of Costs}
\label{supp:application-verification}

The strict local minimum in Proposition~{\ifECMA 14\else F.1\fi} has a highway coefficient satisfying
\begin{align}
0.40104748449567<\beta<0.40104748449569.
\label{eq:application-local-beta}
\end{align}
Its transport cost lies in $(0.00058263725949642,0.00058263725949643)$.

We include the intercept when establishing local optimality. Let $S_g\in\R^{10\times11}$ select coordinates $2,\ldots,11$, let $T\in\R^{8\times11}$ select coordinates $1,5,6,\ldots,11$, and let $P_f\in\R^{10\times11}$ select every coordinate except the fifth. Then, we have $\theta=T\trans\vartheta$. Let $\lambda\in\R^{10}$ denote the multiplier for the sample moment conditions. For $c_{\mathrm{app}}=S_g\trans\lambda$, $u_{\lambda}=P_fc_{\mathrm{app}}$, and $v_{\vartheta}=P_f\theta$, the Hessian of the Lagrangian term with respect to the ten free coordinates of $z_i$ is
\begin{align}
A_f=I_{10}+u_{\lambda}v_{\vartheta}\trans+v_{\vartheta}u_{\lambda}\trans.
\label{eq:application-inner-hessian}
\end{align}
If $A_f$ is positive definite, the unique global minimizer of that term satisfies
\begin{align}
P_fz_i=A_f^{-1}\left(P_fx_i+(y_i-\alpha_{\mathrm{app}})u_{\lambda}-c_{\mathrm{app},5}v_{\vartheta}\right),\qquad z_{i5}=1.
\label{eq:application-linear-adjustment}
\end{align}
In the following definitions, $z_i$ is the function of $(\vartheta,\lambda)$ in \eqref{eq:application-linear-adjustment}. Set $r_{\vartheta,i}=Tz_i$, $w_{g,i}=S_gz_i$, and
\begin{align}
H_{i,f}&=e_{i,\mathrm{app}}S_gP_f\trans-w_{g,i}v_{\vartheta}\trans,\qquad
B_i=u_{\lambda}r_{\vartheta,i}\trans+(c_{\mathrm{app}}\trans z_i)P_fT\trans,\nonumber\\
\mathcal E&=\widehat{\E}[B_i\trans A_f^{-1}B_i],\qquad
\mathcal F=\widehat{\E}[B_i\trans A_f^{-1}H_{i,f}\trans],\nonumber\\
M_{\mathrm{app}}&=\widehat{\E}[H_{i,f}A_f^{-1}H_{i,f}\trans],\qquad
G_{\mathrm{app}}=-\widehat{\E}[w_{g,i}r_{\vartheta,i}\trans],\qquad K=\mathcal F-G_{\mathrm{app}}\trans.
\label{eq:application-derivative-matrices}
\end{align}
For $\xi=(\vartheta\trans,\lambda\trans)\trans\in\R^{18}$, define the vector of first-order conditions $f(\xi)$ by its components $f_\vartheta(\xi)=\widehat{\E}[r_{\vartheta,i}(c_{\mathrm{app}}\trans z_i)]$ and $f_\lambda(\xi)=\widehat{\E}[w_{g,i}e_{i,\mathrm{app}}]$, in that order. Direct differentiation gives
\begin{align}
\partial_{\xi\trans}f(\xi)
=\begin{pmatrix}-\mathcal E&K\\-K\trans&M_{\mathrm{app}}\end{pmatrix}.
\label{eq:application-kkt-jacobian}
\end{align}
If $M_{\mathrm{app}}$ is positive definite, the moment equations determine a smooth $\lambda(\vartheta)$ near a root. Positive definiteness of $A_f$, together with feasibility, shows that the resulting adjusted values uniquely minimize the transport cost at each nearby fixed $\vartheta$. The gradient of this minimized cost is $\widehat{\E}[r_{\vartheta,i}(c_{\mathrm{app}}\trans z_i)]$, and its Hessian at the root is
\begin{align}
\mathcal H_\vartheta=KM_{\mathrm{app}}^{-1}K\trans-\mathcal E.
\label{eq:application-profile-hessian}
\end{align}
Positive definiteness of $\mathcal H_\vartheta$ therefore gives a strict local minimum with respect to all eight regression coefficients, including the intercept. Because the minimizing adjusted values are unique and the Lagrangian is strictly convex with respect to those values, the corresponding point is a strict local minimum in $(\vartheta,z_1,\ldots,z_n)$.

The code in our replication package verifies a root of $f$, the positive definiteness of $A_f$, $M_{\mathrm{app}}$, and $\mathcal H_\vartheta$, and the bounds on the coefficient and the cost by a computer-assisted proof. It converts the binary64 inputs to exact rational numbers and computes their raw moments exactly. From a 90-digit approximation, it defines an 85-significant-digit decimal approximation $\xi_0$ to a root of $f$ and a decimal matrix $R_f$ approximating the inverse Jacobian, and it then treats the entries of both as exact rational numbers. For $r_{\mathcal B}=10^{-40}$, let $\mathcal B=\{\xi:\norm{\xi-\xi_0}_\infty\leq r_{\mathcal B}\}$, where the matrix infinity norm below is the induced operator norm. Outward-rounded interval arithmetic at 70 decimal digits, followed by exact rational calculations with the lower and upper bounds of each interval, gives
\begin{align}
\eta_{\mathcal B}:=\norm{R_ff(\xi_0)}_\infty&<1.386\times10^{-64},\notag\\*
\rho_{\mathcal B}:=\sup_{\xi\in\mathcal B}\norm{I_{18}-R_f\partial_{\xi\trans}f(\xi)}_\infty
&<9.448\times10^{-32}.
\label{eq:application-contraction-bounds}
\end{align}
The lower and upper bounds for each coordinate in $\mathcal B$ are formed by exact decimal addition and subtraction before outward rounding. Exact rational comparisons of the interval upper bounds establish $\eta_{\mathcal B}+\rho_{\mathcal B} r_{\mathcal B}<r_{\mathcal B}$ and $\rho_{\mathcal B}<1$. The map $\xi\mapsto\xi-R_ff(\xi)$ is therefore a contraction from $\mathcal B$ into itself and has a unique fixed point in this box. The bound $\rho_{\mathcal B}<1$ also implies that $R_f$ is invertible, so this fixed point is a root of $f$.

Throughout the same box, interval $\mathrm{LDL}\trans$ decompositions bound every diagonal pivot below by $0.9990$ for $A_f$, $0.02029$ for $M_{\mathrm{app}}$, and $0.0002967$ for $\mathcal H_\vartheta$, so all three matrices are positive definite. Interval evaluation of the coefficient and cost gives \eqref{eq:application-local-beta} and the cost interval stated above.

For the comparison point, take the rational vector
\begin{align}
a^{\text{low}}=\big(&0.39934644946723247,\ 1.2446051493512462,\ -0.6553296708867123,\nonumber\\*
&-0.5730703810605483,\ 1.1383122154639935,\ -1.2546166158167071,\nonumber\\*
&0.5820575701547314\big)\trans.
\label{eq:application-low-coefficients}
\end{align}
The entries after the highway coefficient correspond, in order, to log employment, log market access, log population in 1920, 1950, and 2000, and log manufacturing share in 2003. The code constructs a nonzero rational vector $s_{e,0}\in\boldsymbol 1^\perp$ from an approximate eigenvector in \eqref{eq:application-rayleigh-profile}. Replace $s_es_e\trans$ in \eqref{eq:application-reconstruction} by $s_{e,0}s_{e,0}\trans/(s_{e,0}\trans s_{e,0})$ and use $a=a^{\text{low}}$. No square root is needed, so all adjusted values and the intercept are rational. The reconstruction is feasible for every such $s_{e,0}$, independently of its accuracy as an eigenvector. Exact rational evaluation gives zero for every sample moment and transport cost less than $0.000364$. This cost is strictly smaller than the interval lower bound for the strict local minimum, proving Proposition~{\ifECMA 14\else F.1\fi}.

\end{document}